\documentclass[11pt]{article}

\usepackage[margin=1in]{geometry}
\usepackage[pdftex]{graphicx}
\usepackage[update,prepend]{epstopdf}
\usepackage[justification=justified,singlelinecheck=false,font=small,labelfont=bf]{caption}
\usepackage{subcaption}
\usepackage{booktabs} 
\usepackage{tabularx}
\newcolumntype{Y}{>{\centering\arraybackslash}X} 
\usepackage{amsmath}
\allowdisplaybreaks[4]
\usepackage{amssymb}
\usepackage{amsfonts}
\usepackage{amsthm}
\usepackage[title]{appendix}
\usepackage{dsfont}
\usepackage{bold-extra}
\usepackage{titlesec}
\usepackage[title]{appendix}
\usepackage{cases}
\usepackage{enumerate}
\usepackage{verbatim}
\usepackage{xcolor}
\usepackage{mathtools}

\DeclarePairedDelimiter\floor{\lfloor}{\rfloor}
\allowdisplaybreaks[4]
\usepackage{hyperref}
\hypersetup{pdfpagemode=UseNone} 
\hypersetup{pdfstartview={XYZ null null 0.85}} 

\newcommand{\RR}{\mathbb{R}}

\DeclareMathOperator{\E}{\mathbb{E}}
\DeclareMathOperator{\Prob}{\mathbb{P}}
\DeclareMathOperator{\Ind}{\mathds{1}}

\DeclareMathOperator{\sgn}{\text{sgn}}

\newtheorem{theorem}{Theorem}[section]
\newtheorem{proposition}[theorem]{Proposition}
\newtheorem{lemma}[theorem]{Lemma}
\newtheorem{remark}[theorem]{Remark}
\newtheorem{assumption}[theorem]{Assumption}
\newtheorem{corollary}[theorem]{Corollary}
\numberwithin{equation}{section}

\title{Strong convergence rates for Euler approximations to\\ a class of stochastic path-dependent volatility models}

\author{\scshape{andrei cozma}\thanks{\footnotesize\scshape{Mathematical Institute, University of Oxford, Oxford OX2 6GG, United Kingdom}\protect\\ \footnotesize\scshape{The first author gratefully acknowledges financial support from the Engineering and Physical Sciences Research Council (research grant EPSRC EP/N509711/1).}} \and \scshape{christoph reisinger\footnotemark[1]}}

\date{}
\begin{document}
\maketitle

\begin{abstract}
\noindent
We consider a class of stochastic path-dependent volatility models where the stochastic volatility, whose square follows the Cox--Ingersoll--Ross model, is multiplied by a (leverage) function of the spot process, its running maximum, and time. We propose a Monte Carlo simulation scheme which combines a log-Euler scheme for the spot process with the full truncation Euler scheme or the backward Euler--Maruyama scheme for the squared stochastic volatility component. Under some mild regularity assumptions and a condition on the Feller ratio, we establish the strong convergence with order 1/2 (up to a logarithmic factor) of the approximation process up to a critical time. The model studied in this paper contains as special cases Heston-type stochastic-local volatility models, the state-of-the-art in derivative pricing, and a relatively new class of path-dependent volatility models. The present paper is the first to prove the convergence of the popular Euler schemes with a positive rate, which is moreover consistent with that for Lipschitz coefficients and hence optimal.

\vspace{1em}

\vspace{.6em}
\end{abstract}

\section{Introduction}\label{sec:intro}

The two major families of option pricing models are local volatility (LV) (e.g., \cite{Dupire:1994}) and stochastic volatility (SV) (e.g., \cite{Heston:1993}). LV models are flexible enough to perfectly replicate the market prices of vanilla options, whereas SV models generate much richer and more realistic spot-vol dynamics. The class of stochastic-local volatility (SLV) models introduced in \cite{Jex:1999,Lipton:2002a,Lipton:2002b} contain a stochastic volatility component and a local volatility component (the leverage function), and combine advantages of the two. According to \cite{Ren:2007,Tian:2015,Stoep:2014}, they allow for a better calibration to vanilla options and improve the pricing and risk-management performance. SLV models were recently referred to in \cite{Lipton:2014} as the de facto standard for pricing foreign exchange (FX) options.

European options are actively traded on many asset classes, including FX, equities and commodities. Barrier options are also actively traded in these markets, and especially in FX markets. Their popularity can be explained by two key factors. First, they are useful in limiting the risk exposure of an investor. Second, they offer additional flexibility and can match an investor's view on the market for a lower price than a European option. Barrier options, and in particular no-touch options, are so heavily traded that they are no longer considered exotic options. Hence, a pricing model that allows a perfect calibration to both European and no-touch options is desirable. While the prices of European options depend only on the final distribution of the underlying spot process (e.g., a stock price or a spot FX rate), the prices of no-touch options depend on the entire distribution of the underlying throughout the duration of the contract. Path-dependent volatility (PDV) models (see, e.g., \cite{Guyon:2014} and the references therein) assume that the volatility depends on the path of the underlying through the current value of the spot and a finite number of path-dependent variables, like the running or the moving average, the running maximum or minimum etc. PDV models are complete, can be perfectly calibrated to both vanilla and no-touch options (e.g., \cite{Mariapragassam:2017}), and can produce rich implied volatility dynamics. Furthermore, according to \cite{Brunick:2013}, the joint distribution of the spot process and any path-dependent quantity of any SV or SLV model agrees with that of a suitably chosen PDV model. As a consequence, there is always a PDV model that produces the same prices of both vanilla and exotic options and that can reproduce SLV spot-vol dynamics.

Stochastic path-dependent volatility (SPDV) models were briefly discussed in \cite{Guyon:2014} as a generalization of PDV models. Although incomplete, they generate richer spot-vol dynamics, and include as special cases both SLV and PDV models. In this paper, we consider a Heston-type SPDV model because the Cox--Ingersoll--Ross (CIR) process \cite{Cox:1985} for the squared stochastic volatility is widely used in the industry due to its desirable properties, such as mean-reversion, non-negativity and analytical tractability, which allows for a fast calibration of the stochastic volatility parameters. Furthermore, we introduce path-dependency only through the running maximum of the underlying spot process, which allows for an exact calibration to both vanilla and no-touch options \cite{Mariapragassam:2017}.

The SPDV model is non-affine and hence a closed-form solution to the European option valuation problem is not available. Therefore, we use Monte Carlo simulation methods \cite{Glasserman:2004} -- which can handle path-dependent features easily -- and approximate the solution to the stochastic differential equation (SDE) using an explicit or implicit Euler or Milstein discretization. Weak convergence is important when estimating expectations of (discounted) payoffs. Strong convergence plays a crucial role in multilevel Monte Carlo methods \cite{Giles:2008,Giles:2009,Kebaier:2005} and may be required for some complex path-dependent derivatives. Furthermore, pathwise convergence follows automatically \cite{Kloeden:2007}.

The usual theorems in \cite{Kloeden:1999} on the convergence of numerical simulations assume that the drift and diffusion coefficients are globally Lipschitz continuous and satisfy a linear growth condition, whereas \cite{Higham:2002} extended the analysis to locally Lipschitz SDEs. The standard convergence theory does not apply to the present work because of the explicit dependence of the drift and diffusion coefficients on the running maximum and also since the square-root diffusion coefficient of the CIR process is not Lipschitz. Strong and weak divergence of Euler approximations to SDEs with superlinearly growing coefficients was proved in \cite{Hutzenthaler:2011}. Furthermore, a considerable amount of research has recently been devoted to proving that approximation schemes for some multi-dimensional SDEs with infinitely often differentiable and globally bounded coefficients converge arbitrarily slowly \cite{Gerencser:2017,Hairer:2015,Jentzen:2016,Muller-Gronbach:2017,Yaroslavtseva:2016}. In particular, it was shown in \cite{Gerencser:2017} that for any arbitrarily slow speed of convergence, there exists a 2-dimensional SDE with smooth and bounded coefficients such that no approximation method based on finitely many observations of the driving Brownian motion can converge in $L^{1}$ to the solution faster than the given speed of convergence. These slow convergence phenomena raise the natural question of whether Euler approximations to the class of SPDV models considered in this paper also converge in the strong sense arbitrarily slowly, if at all.

The literature on the convergence of Monte Carlo methods under stochastic volatility is scarce. The Heston stochastic volatility model was considered in \cite{Higham:2005} and the strong convergence without a rate as well as the weak convergence for bounded payoffs were derived for a stopped Euler scheme with a reflection fix in the diffusion coefficient. For the log-Heston model, the convergence in $L^{p}$ with order 1/2 up to a logarithmic factor was established in \cite{Kloeden:2012} when the Euler scheme and the backward (drift-implicit) Euler--Maruyama (BEM) scheme are employed in the discretization of the (log-)spot process and its squared volatility, respectively. For the Heston model, \cite{Altmayer:2017} proved the weak convergence with order 1 of a log-Euler (LE) scheme for the spot process and a drift-implicit Milstein scheme for its squared volatility. Moreover, using the full truncation Euler (FTE) or the BEM scheme instead to discretize the squared volatility, the convergence in $L^{p}$ with order 1/2 up to a logarithmic factor can easily be deduced by using some strong convergence results of \cite{Cozma:2017c,Dereich:2012} together with a recent moment bound result of \cite{Cozma:2016a}. A hybrid Heston-type stochastic-local volatility model with stochastic short interest rates was considered in \cite{Cozma:2016a} and the strong convergence without a rate as well as the weak convergence for vanilla and exotic options were derived when the spot process is discretized via the LE scheme and its squared volatility and the short rates are discretized via the FTE scheme. The convergence rate, however, remained an open question and this paper is the first to address it.

In this work, we prove the strong convergence in $L^{p}$ with order 1/2 up to a logarithmic factor of the Monte Carlo method with the LE scheme for the spot process and the (explicit) FTE scheme proposed in \cite{Lord:2010} or the (implicit) BEM scheme proposed in \cite{Alfonsi:2005} for the squared volatility. The FTE scheme is arguably the most widely used scheme in practice because it preserves the positivity of the original process, is easy to implement and, perhaps most importantly, is found empirically to produce the smallest bias of all explicit Euler schemes with different fixes at the boundary \cite{Lord:2010}. The BEM scheme is often encountered in the finance literature and its convergence properties are well-understood \cite{Alfonsi:2013,Dereich:2012,Hutzenthaler:2014,Neuenkirch:2014}. Hence, we obtain the optimal strong convergence rate for the numerical approximation of SDEs with globally Lipschitz coefficients \cite{Hofmann:2001,Muller-Gronbach:2002}. As a consequence, the Euler discretization of the spot process also converges with weak order 1/2 (up to a logarithmic factor), which is optimal because the Euler scheme for the running maximum converges with weak order of at most 1/2 (see, e.g., \cite{Asmussen:1995,Glasserman:2004}) rather than the weak order 1 typical for SDEs with smooth coefficients.

In summary, to the best of our knowledge, this paper is the first to establish a positive strong convergence rate for Euler approximations to models with: (1) path-dependent volatility dynamics; (2) local and stochastic volatility dynamics, even without the path-dependency.

The remainder of this paper is structured as follows. In Section~\ref{sec:setup}, we discuss the model and the postulated assumptions. Then, we define the simulation scheme and discuss the main theorem. In Section~\ref{sec:vol}, we present some auxiliary integrability and convergence results for the CIR process and its FTE and BEM discretizations. In Section~\ref{sec:spot}, we prove the strong convergence with a rate of the approximated spot process. In Section~\ref{sec:numerics}, we conduct numerical tests for the strong and weak convergence rates that validate and complement our theoretical findings. Section~\ref{sec:conclusion} contains a short discussion. Finally, detailed proofs of some technical results are given in the Appendix.

\section{Set-up and main result}\label{sec:setup}

\subsection{Model assumptions}\label{subsec:model}

Consider a filtered probability space $\big(\Omega,\mathcal{F},(\mathcal{F}_{t})_{t\geq0},\mathbb{P}\big)$ satisfying the usual conditions and let $W^{s}=\big(W^{s}_{t}\big)_{t\geq0}$ and $W^{v}=\big(W^{v}_{t}\big)_{t\geq0}$ be one-dimensional standard $\mathcal{F}_{t}$-adapted Brownian motions. We study the model
\begin{align}\label{eq2.1.1}
	\begin{dcases}
	dS_{t} = \mu(t,S_t,M_t)S_{t}dt + \sqrt{v_{t}}\hspace{.5pt}\sigma(t,S_t,M_t)S_{t}\hspace{1pt}dW^{s}_{t},\hspace{.75em} S_{0}>0, \\[2pt]
	dv_{t} \hspace{1pt} = k(\theta-v_{t})dt + \xi\sqrt{v_{t}}\,dW^{v}_{t},\hspace{.75em} v_{0}>0, \\[1pt]
	M_{t} \hspace{1.5pt}= \sup_{u\in[0,t]} S_u,
	\end{dcases}
\end{align}
where $d\langle W^{s},W^{v}\rangle_{t}=\rho\hspace{.5pt}dt$ with $\rho\in(-1,1)$. For a fixed time horizon $T>0$, let
\begin{equation}\label{eq2.1.2}
\mu:[0,T]\!\times\!\big\{(x,y)\in\mathbb{R}^{2}_{+} \,|\, S_{0}\vee x\leq y\big\}\rightarrow\mathbb{R}
\end{equation}
and
\begin{equation}\label{eq2.1.3}
\sigma:[0,T]\!\times\!\big\{(x,y)\in\mathbb{R}^{2}_{+} \,|\, S_{0}\vee x\leq y\big\}\rightarrow\mathbb{R}_{+}.
\end{equation}
On the one hand, when the running maximum component vanishes, i.e., when $\mu(t,S_t,M_t)=\mu(t,S_t)$ and $\sigma(t,S_t,M_t)=\sigma(t,S_t)$, the SPDV model \eqref{eq2.1.1} collapses to a SLV model. The SPDV model further reduces to a LV model if the stochastic volatility component also vanishes, i.e., if we take $\xi$ to be zero, or to a SV model if we take $\mu$ and $\sigma$ to be constant.

On the other hand, when the stochastic volatility component vanishes, i.e., when $\xi$ is zero, the SPDV model collapses to a PDV model, which can also be regarded as the Markovian projection of an It\^o process onto the spot and its running maximum \cite{Brunick:2013}.

In this paper, we work under the following model assumptions:
\begin{assumption}\label{Asm2.1}
The drift and diffusion functions $\mu$ and $\sigma$ are bounded, i.e., there exist non-negative constants $\mu_{max}$ and $\sigma_{max}$ such that, for all $t\in [0,T]$ and $0\leq x\leq S_{0}\vee x\leq y$, we have
\begin{equation}\label{eq2.0.1}
\left|\mu(t,x,y)\right| \leq \mu_{max}
\end{equation}
and
\begin{equation}\label{eq2.0.2}
0 \leq \sigma(t,x,y) \leq \sigma_{max}.
\end{equation}
\end{assumption}
\begin{assumption}\label{Asm2.2}
The drift and diffusion functions $\mu$ and $\sigma$ are bounded and piecewise 1/2-H\"older continuous in time, respectively, and Lipschitz continuous in log-spot and log-running maximum, i.e., there exist $N_{T}\in\mathbb{N}$ and non-negative constants $C_{\mu,t}$, $C_{\mu,x}$, $C_{\mu,m}$, $C_{\sigma,t}$, $C_{\sigma,x}$, $C_{\sigma,m}$ and $(C_{\sigma,t,j})_{1\leq j\leq N_{T}}$ such that, for all $t_{1}, t_{2} \in [0,T]$, $0<x_{1}\leq S_{0}\vee x_{1}\leq y_{1}$ and $0<x_{2}\leq S_{0}\vee x_{2}\leq y_{2}$, we have
\begin{align}\label{eq2.0.3}
\left|\mu(t_{1},x_{1},y_{1})-\mu(t_{2},x_{2},y_{2})\right| &\leq C_{\mu,t}\Ind_{t_{1}\neq t_{2}} +\hspace{2pt} C_{\mu,x}\left|\log(x_{1})-\log(x_{2})\right| \nonumber\\[3pt]
&+ C_{\mu,m}\left|\log(y_{1})-\log(y_{2})\right|
\end{align}
and
\begin{align}\label{eq2.0.4}
\left|\sigma(t_{1},x_{1},y_{1})-\sigma(t_{2},x_{2},y_{2})\right| &\leq C_{\sigma,t}\sqrt{\left|t_{1}-t_{2}\right|} + \sum_{j=1}^{N_{T}}{C_{\sigma,t,j}\Ind_{t_{1}\wedge\hspace{.5pt}t_{2}\hspace{.5pt}<\frac{jT}{N_{T}}\leq\hspace{.5pt}t_{1}\vee\hspace{.5pt}t_{2}}} \nonumber\\[1pt]
&+ C_{\sigma,x}\left|\log(x_{1})-\log(x_{2})\right| + C_{\sigma,m}\left|\log(y_{1})-\log(y_{2})\right|.
\end{align}
\end{assumption}

As an aside, note that jumps in the function $\sigma$ do not have to be equally spaced as long as they occur at the time discretization nodes.

\begin{remark}\label{Rem2.3}
If Assumptions \ref{Asm2.1} and \ref{Asm2.2} hold, then for the purpose of this paper we choose the smallest non-negative constants possible. In particular,
\begin{align}
\sigma_{max} &= \sup\Big\{\sigma(t,x,y) \,\big|\; t\in[0,T],\, 0\leq x\leq S_{0}\vee x\leq y\Big\}, \label{eq4.1.23}\\[4pt]
C_{\sigma,x} &= \sup\bigg\{\frac{\left|\sigma(t,x_{1},y)-\sigma(t,x_{2},y)\right|}{\left|\log(x_{1})-\log(x_{2})\right|} \,\big|\; t\in[0,T],\, 0<x_{1}<x_{2}\leq S_{0}\vee x_{2}\leq y\bigg\}, \label{eq4.1.27}\\[4pt]
C_{\sigma,m} &= \sup\bigg\{\frac{\left|\sigma(t,x,y_{1})-\sigma(t,x,y_{2})\right|}{\left|\log(y_{1})-\log(y_{2})\right|} \,\big|\; t\in[0,T],\, 0<x\leq S_{0}\vee x\leq y_{1}<y_{2}\bigg\}. \label{eq4.1.29}
\end{align}
\end{remark}

We present a set of stronger assumptions and discuss them from a practical point of view.
\begin{assumption}\label{Asm2.4}
The drift and diffusion functions $\mu$ and $\sigma$ are constant outside a bounded interval, i.e., we can find $0\leq S_{min}<S_{0}<S_{max}$ such that, for all $t \in [0,T]$ and $0\leq x\leq S_{0}\vee x\leq y$, we have
\begin{equation}\label{eq2.1.4}
\mu(t,x,y)=\mu(t,S_{min}\vee x\wedge S_{max},S_{min}\vee y\wedge S_{max})
\end{equation}
and
\begin{equation}\label{eq2.1.5}
\sigma(t,x,y)=\sigma(t,S_{min}\vee x\wedge S_{max},S_{min}\vee y\wedge S_{max}).
\end{equation}
\end{assumption}
\begin{assumption}\label{Asm2.5}
The drift and diffusion functions $\mu$ and $\sigma$ are bounded and piecewise 1/2-H\"older continuous in time, respectively, and Lipschitz continuous in spot and running maximum, i.e., there exist non-negative constants $C_{\mu,S}$, $C_{\mu,M}$, $C_{\sigma,S}$ and $C_{\sigma,M}$ such that, for all $t_{1}, t_{2} \in [0,T]$, $0\leq x_{1}\leq S_{0}\vee x_{1}\leq y_{1}$ and $0\leq x_{2}\leq S_{0}\vee x_{2}\leq y_{2}$, we have
\begin{equation}\label{eq2.1.6}
\left|\mu(t_{1},x_{1},y_{1})-\mu(t_{2},x_{2},y_{2})\right| \leq C_{\mu,t}\Ind_{t_{1}\neq t_{2}} +\hspace{2pt} C_{\mu,S}\left|x_{1}-x_{2}\right| + C_{\mu,M}\left|y_{1}-y_{2}\right|
\end{equation}
and
\begin{align}\label{eq2.1.7}
\left|\sigma(t_{1},x_{1},y_{1})-\sigma(t_{2},x_{2},y_{2})\right| &\leq C_{\sigma,t}\sqrt{\left|t_{1}-t_{2}\right|} + \sum_{j=1}^{N_{T}}{C_{\sigma,t,j}\Ind_{t_{1}\wedge\hspace{.5pt}t_{2}\hspace{.5pt}<\frac{jT}{N_{T}}\leq\hspace{.5pt}t_{1}\vee\hspace{.5pt}t_{2}}} \nonumber\\[1pt]
&+ C_{\sigma,S}\left|x_{1}-x_{2}\right| + C_{\sigma,M}\left|y_{1}-y_{2}\right|.
\end{align}
\end{assumption}
When modelling asset prices or spot FX rates, the drift function $\mu$ is usually a combination of deterministic short interest rates and dividend yields, such that $\mu(t,S_t,M_t)=\mu(t)$ satisfies the above assumptions. For the diffusion function $\sigma$ to be consistent with European call and put prices, it has to be given by the ratio between a calibrated Brunick--Shreve volatility and the square-root of the conditional expectation of the squared stochastic volatility \cite{Brunick:2013,Mariapragassam:2017}. In case of no running maximum component, the leverage function $\sigma$ that is consistent with vanilla prices is given by the ratio between a calibrated Dupire local volatility and the square-root of the conditional expectation of the squared stochastic volatility \cite{Guyon:2014,Guyon:2011,Ren:2007}. In practice, the leverage function is defined on a grid of points (20 points per year in time and 30 points in space usually suffice for an acceptable calibration error \cite{Cozma:2017a,Guyon:2011}), interpolated flat-forward in time and using cubic splines in spot, and extrapolated flat outside an interval. Hence, there is no significant loss of generality from a practical point of view in making these two assumptions.

\begin{proposition}\label{Prop2.6}
Suppose that Assumptions \ref{Asm2.4} and \ref{Asm2.5} are satisfied. Then Assumptions \ref{Asm2.1} and \ref{Asm2.2} are also satisfied.
\end{proposition}
\begin{proof}
See Appendix \ref{sec:aux1}.
\end{proof}

\subsection{Simulation scheme}\label{subsec:scheme}

Let $N\in\mathbb{N}$ be a multiple of $N_{T}$ and consider a uniform grid
\begin{equation}\label{eq2.2.1}
T = N\delta t,\hspace{1em} t_{n}=n\delta t,\hspace{1em} \forall\hspace{1pt} n\in\{0,1,...,N\}.
\end{equation}
We use either the full truncation Euler scheme from \cite{Lord:2010} or the drift-implicit (square-root) Euler scheme, also known as the backward Euler--Maruyama scheme, proposed in \cite{Alfonsi:2005} in order to discretize the squared volatility $v$. For the FTE discretization, we introduce the discrete-time auxiliary process
\begin{equation}\label{eq2.2.2}
\tilde{v}_{t_{n+1}} = \tilde{v}_{t_{n}} + k(\theta-\tilde{v}_{t_{n}}^{+})\delta t + \xi\sqrt{\tilde{v}_{t_{n}}^{+}}\,\delta W^{v}_{t_{n}},\hspace{.75em} \tilde{v}_{0}=v_{0},
\end{equation}
where $v^{+} = \max\left(0,v\right)$ and $\delta W^{v}_{t_{n}} = W^{v}_{t_{n+1}} - W^{v}_{t_{n}}$, its continuous-time interpolation
\begin{equation}\label{eq2.2.3}
\tilde{v}_{t} = \tilde{v}_{t_{n}} + k(\theta-\tilde{v}_{t_{n}}^{+})(t-t_{n}) + \xi\sqrt{\tilde{v}_{t_{n}}^{+}}\big(W^{v}_{t}-W^{v}_{t_{n}}\big)
\end{equation}
as well as the non-negative piecewise constant process
\begin{equation}\label{eq2.2.4}
\bar{v}_{t}=\tilde{v}_{t_{n}}^{+},
\end{equation}
whenever $t \in [t_{n},t_{n+1})$. For the BEM discretization, assuming the boundary point 0 is inaccessible (i.e., $2k\theta\geq\xi^{2}$) and using a Lamperti transformation $y=\sqrt{v}$, we deduce that
\begin{equation}\label{eq2.2.5}
dy_{t} = \big(\alpha y_{t}^{-1}+\beta y_{t}\big)dt + \gamma\hspace{1pt}dW^{v}_{t},
\end{equation}
where
\begin{equation}\label{eq2.2.6}
\alpha = \frac{4k\theta-\xi^{2}}{8}\hspace{1pt},\hspace{.5em} \beta = -\hspace{1pt}\frac{k}{2} \hspace{.5em}\text{ and }\hspace{.5em} \gamma = \frac{\xi}{2}\hspace{1pt}.
\end{equation}
We introduce the discrete-time auxiliary process
\begin{equation}\label{eq2.2.7}
\tilde{y}_{t_{n+1}} = \tilde{y}_{t_{n}} + \big(\alpha\tilde{y}_{t_{n+1}}^{-1} + \beta\tilde{y}_{t_{n+1}}\big)\delta t + \gamma\delta W^{v}_{t_{n}},\hspace{.75em} \tilde{y}_{0}=y_{0},
\end{equation}
as well as the piecewise constant processes
\begin{equation}\label{eq2.2.8}
\bar{y}_{t} = \tilde{y}_{t_{n}} \hspace{.5em}\text{ and }\hspace{.5em} \bar{v}_{t} = \tilde{y}_{t_{n}}^{2},
\end{equation}
whenever $t \in [t_{n},t_{n+1})$. If $4k\theta>\xi^{2}$, then $\alpha>0$ and, since $\beta<0$, \eqref{eq2.2.7} has the unique positive solution
\begin{equation}\label{eq2.2.9}
\tilde{y}_{t_{n+1}} = \frac{\tilde{y}_{t_{n}}+\gamma\hspace{1pt}\delta W^{v}_{t_{n}}}{2(1-\beta\delta t)} + \sqrt{\frac{(\tilde{y}_{t_{n}}+\gamma\hspace{1pt}\delta W^{v}_{t_{n}})^{2}}{4(1-\beta\delta t)^{2}}+\frac{\alpha\hspace{.5pt}\delta t}{1-\beta\delta t}}\hspace{1pt}.
\end{equation}
Note that unlike in \cite{Alfonsi:2013,Dereich:2012,Neuenkirch:2014}, it is critical that we employ a piecewise constant continuous-time interpolation $\bar{v}$ for the squared volatility $v$ because we only simulate increments of the Brownian driver $W^{s}$ and hence the diffusion coefficient of the spot process needs to be constant in between time nodes. We employ an Euler--Maruyama scheme to discretize the log-spot process $x=(x_{t})_{t\geq0}$, where $x_{t}=\log(S_{t})$, and we define for convenience the log-running maximum $m=(m_{t})_{t\geq0}$, where $m_{t}=\log(M_{t})=\sup_{u\in[0,t]}x_{u}$. Let $\bar{x}$ be the approximated log-spot process, then the discrete method reads:
\begin{align}
\bar{x}_{t_{n+1}} &= \bar{x}_{t_{n}} + \int_{t_{n}}^{t_{n+1}}{\mu\big(u,e^{\bar{x}_{t_{n}}},e^{\bar{m}_{t_{n}}}\big)du} - \frac{1}{2}\hspace{1pt}\sigma^{2}\big(t_{n},e^{\bar{x}_{t_{n}}},e^{\bar{m}_{t_{n}}}\big)\bar{v}_{t_{n}}\delta t \nonumber\\[2pt]
&+ \sigma\big(t_{n},e^{\bar{x}_{t_{n}}},e^{\bar{m}_{t_{n}}}\big)\sqrt{\bar{v}_{t_{n}}}\hspace{1pt}\delta W^{s}_{t_{n}},\hspace{.75em} \bar{x}_{0}=x_{0}, \label{eq2.2.10}\\[3pt]
\bar{m}_{t_{n+1}} &= \max_{0\leq i\leq n+1}\bar{x}_{t_{i}},\hspace{.75em} \bar{m}_{0}=x_{0}. \label{eq2.2.10b}
\end{align}
The continuous-time approximation is
\begin{equation}\label{eq2.2.11}
\bar{x}_{t} = \bar{x}_{t_{n}} + \int_{t_{n}}^{t}{\bar{\mu}\big(u,e^{\bar{x}_{u}},e^{\bar{m}_{u}}\big)du} - \frac{1}{2}\hspace{1pt}\bar{\sigma}^{2}\big(t,e^{\bar{x}_{t}},e^{\bar{m}_{t}}\big)\bar{v}_{t}(t-t_{n})+ \bar{\sigma}\big(t,e^{\bar{x}_{t}},e^{\bar{m}_{t}}\big)\sqrt{\bar{v}_{t}}\hspace{1pt}\big(W^{s}_{t}-W^{s}_{t_{n}}\big),
\end{equation}
whenever $t\in[t_{n},t_{n+1})$, where $\bar{\mu}\big(t,e^{\bar{x}_{t}},e^{\bar{m}_{t}}\big)=\mu\big(t,e^{\bar{x}_{t_{n}}},e^{\bar{m}_{t_{n}}}\big)$ and $\bar{\sigma}\big(t,e^{\bar{x}_{t}},e^{\bar{m}_{t}}\big)=\sigma\big(t_{n},e^{\bar{x}_{t_{n}}},e^{\bar{m}_{t_{n}}}\big)$. Hence,
\begin{equation}\label{eq2.2.12}
\bar{x}_{t} = x_{0} + \int_{0}^{t}{\bar{\mu}\big(u,e^{\bar{x}_{u}},e^{\bar{m}_{u}}\big)du} - \frac{1}{2}\int_{0}^{t}{\bar{\sigma}^{2}\big(u,e^{\bar{x}_{u}},e^{\bar{m}_{u}}\big)\bar{v}_{u}\hspace{1pt}du} + \int_{0}^{t}{\bar{\sigma}\big(u,e^{\bar{x}_{u}},e^{\bar{m}_{u}}\big)\sqrt{\bar{v}_{u}}\hspace{1pt}dW^{s}_{u}}.
\end{equation}
Let $\bar{S}=(\bar{S}_{t})_{t\geq0}$, where $\bar{S}_{t}=e^{\bar{x}_{t}}$, be the continuous-time approximation of $S$, and let $\bar{M}_{t_{n}}=e^{\bar{m}_{t_{n}}}=\max_{0\leq i\leq n}\bar{S}_{t_{i}}$, for all $0\leq n\leq N$. Using It\^o's formula, we obtain
\begin{equation}\label{eq2.2.13}
\bar{S}_{t} = S_{0} + \int_{0}^{t}{\bar{\mu}(u,\bar{S}_{u},\bar{M}_{u})\bar{S}_{u}\hspace{1pt}du} + \int_{0}^{t}{\bar{\sigma}(u,\bar{S}_{u},\bar{M}_{u})\sqrt{\bar{v}_{u}}\,\bar{S}_{u}\hspace{1pt}dW^{s}_{u}}.
\end{equation}
We prefer the log-Euler scheme to the standard Euler scheme to discretize the spot process because the former preserves positivity and produces no discretization bias in the spot direction when $\mu$ is deterministic and $\sigma$ is constant, which is desirable because the drift function may be discontinuous.

\subsection{The main theorem}\label{subsec:main}

Before we state the main result, we introduce some necessary notations. Throughout this paper, we use a superscript $*\in\{\text{FTE},\text{BEM}\}$ to differentiate between the two discretization schemes for the squared volatility process. Let the Feller ratio be
\begin{equation}\label{eq2.3.1}
\nu = \frac{2k\theta}{\xi^{2}}\hspace{1pt}.
\end{equation}
For brevity, define
\begin{equation}\label{eq2.3.2a}
\nu^{\scalebox{0.6}{\text{FTE}}}=2+\sqrt{3},\hspace{1em} \nu^{\scalebox{0.6}{\text{BEM}}}=2,
\end{equation}
and also
\begin{equation}\label{eq2.3.2b}
p^{\scalebox{0.6}{\text{FTE}}}(\nu)=\nu^{-1}(\nu-1)^{2},\hspace{1em} p^{\scalebox{0.6}{\text{BEM}}}(\nu)=\nu.
\end{equation}
Moreover, let $\beta_{0}\approx1.307$ be the unique positive root of
\begin{equation}\label{eq2.3.2c}
\phi_{0}(s) = -e^{\frac{s^{2}}{2}} + s\int_{0}^{s}{e^{\frac{u^{2}}{2}}\,du}\hspace{1pt}.
\end{equation}
First, define
\begin{equation}\label{eq2.3.3}
T^{*}_{x}(p) = \frac{2}{\sqrt{\smash[b]{(\varphi^{*}(p)-k^{2})^{+}}}}\Bigg[\frac{\pi}{2} + \arctan\Bigg(\frac{k}{\sqrt{\smash[b]{(\varphi^{*}(p)-k^{2})^{+}}}}\Bigg)\Bigg],
\end{equation}
where
\begin{align}\label{eq2.3.4}
\varphi^{*}(p) &= \inf_{q\in(p,\hspace{.5pt}p^{*}\hspace{-.5pt})}\tilde{\varphi}(p,q)
\end{align}
and
\begin{align}\label{eq2.3.5}
\tilde{\varphi}(p,q) &= \frac{pq\xi^{2}}{2(q-p)}\Big(\sqrt{(2+\beta_{0}^{2})(C_{\sigma,x}+C_{\sigma,m})^{2}q+2(C_{\sigma,x}+C_{\sigma,m})(2\sigma_{max}-C_{\sigma,x}-C_{\sigma,m})} \nonumber\\[2pt]
&+ \beta_{0}(C_{\sigma,x}+C_{\sigma,m})\sqrt{q}\,\Big)^{2}.
\end{align}
Second, define
\begin{align}\label{eq2.3.6}
T^{\scalebox{0.6}{\text{FTE}}}_{S}(p) =
\begin{dcases}
\frac{4k}{\phi(p)} &\text{if } \phi(p)<4k^{2} \\[2pt]
\frac{1}{\sqrt{\phi(p)}-k} &\text{if } \phi(p)\geq4k^{2}
\end{dcases}
\end{align}
and
\begin{equation}\label{eq2.3.7}
T^{\scalebox{0.6}{\text{BEM}}}_{S}(p) = \frac{1}{\sqrt{\phi(p)}}\hspace{1pt},
\end{equation}
where
\begin{equation}\label{eq2.3.8}
\phi(p)=\xi^{2}\sigma_{max}^{2}\big(p+\sqrt{(p-1)p}\,\big)^{2}.
\end{equation}
Third, define
\begin{equation}\label{eq2.3.9}
T^{*}(p) = \sup_{q\in(2\vee p,\hspace{.5pt}p^{*}\hspace{-.5pt})}\Big[\hspace{1pt}T_{x}^{*}(q) \wedge\hspace{.5pt} T_{S}^{*}\big(pq(q-p)^{-1}\big)\Big],
\end{equation}
with $T_{x}^{*}$ given in \eqref{eq2.3.3} and $T^{\scalebox{0.6}{\text{FTE}}}_{S}$ and $T^{\scalebox{0.6}{\text{BEM}}}_{S}$ in \eqref{eq2.3.6} and \eqref{eq2.3.7}, respectively.

To the best of our knowledge, Theorem \ref{Thm2.7} below is the first result to establish a positive strong convergence rate for Euler approximations to models with local and stochastic volatility dynamics, even without the path-dependency. The proof is postponed to Section \ref{sec:spot}. In Section \ref{sec:numerics}, we briefly examine the critical time $T^{*}$ defined in \eqref{eq2.3.9} with respect to the model parameters in a realistic scenario.
\begin{theorem}\label{Thm2.7}
Suppose that Assumptions \ref{Asm2.1} and \ref{Asm2.2} hold and that $\nu>\nu^{*}$, with $\nu^{*}$ defined in \eqref{eq2.3.2a}. Then for all $1\leq p<p^{*}(\nu)$ and $T<T^{*}(p)$, with $p^{*}$ defined in \eqref{eq2.3.2b} and $T^{*}$ given in \eqref{eq2.3.9}, there exists a constant $C$ such that, for all $N\geq1$,
\begin{equation}\label{eq2.3.10}
\sup_{t\in[0,T]}\E\Big[\big|S_{t}-\bar{S}_{t}\big|^{p}\Big]^{\frac{1}{p}} \leq C\sqrt{\frac{\log(2N)}{N}}\hspace{1pt}.
\end{equation}
\end{theorem}

If the stochastic volatility component vanishes (e.g., take $v_{0}=\theta=1$ and $\xi=0$), then the SPDV model \eqref{eq2.1.1} collapses to a path-dependent volatility model
\begin{align}\label{eq2.3.11}
	\begin{dcases}
	dS^{\scalebox{0.6}{\text{PDV}}}_{t} = \mu(t,S^{\scalebox{0.6}{\text{PDV}}}_t,M^{\scalebox{0.6}{\text{PDV}}}_t)S^{\scalebox{0.6}{\text{PDV}}}_{t}dt + \sigma(t,S^{\scalebox{0.6}{\text{PDV}}}_t,M^{\scalebox{0.6}{\text{PDV}}}_t)S^{\scalebox{0.6}{\text{PDV}}}_{t}\hspace{1pt}dW^{s}_{t},\hspace{.75em} S^{\scalebox{0.6}{\text{PDV}}}_{0}>0, \\[2pt]
	M^{\scalebox{0.6}{\text{PDV}}}_{t} \hspace{1.5pt}= \sup_{u\in[0,t]} S^{\scalebox{0.6}{\text{PDV}}}_u.
	\end{dcases}
\end{align}
Upon noticing from \eqref{eq2.3.1}, \eqref{eq2.3.2b} and \eqref{eq2.3.3} -- \eqref{eq2.3.9} that $\nu=p^{*}(\nu)=\infty$ and $T^{*}(p)=\infty$, for all $p\geq1$, the same argument ensures the strong convergence in $L^{p}$ with order 1/2 (up to a logarithmic factor), for all $p\geq1$, of the corresponding approximation process $\bar{S}^{\scalebox{0.6}{\text{PDV}}}$ defined in \eqref{eq2.2.13}. Corollary \ref{Cor2.8} below is the first result to establish a positive strong convergence rate for Euler approximations to models with path-dependent volatility dynamics, to the best of our knowledge.

\begin{corollary}\label{Cor2.8}
Suppose that Assumptions \ref{Asm2.1} and \ref{Asm2.2} hold. Then for all $p\geq1$, there exists a constant $C$ such that, for all $N\geq1$,
\begin{equation}\label{eq2.3.12}
\sup_{t\in[0,T]}\E\Big[\big|S^{\scalebox{0.6}{\emph{\text{PDV}}}}_{t}-\bar{S}^{\scalebox{0.6}{\emph{\text{PDV}}}}_{t}\big|^{p}\Big]^{\frac{1}{p}} \leq C\sqrt{\frac{\log(2N)}{N}}\hspace{1pt}.
\end{equation}
\end{corollary}

We know from Theorem 10.2.2 in \cite{Kloeden:1999} the strong convergence in $L^{1}$ with order 1/2 of Euler approximations to the LV model
\begin{equation}\label{eq2.3.13}
dS^{\scalebox{0.6}{\text{LV}}}_{t} = \mu(t,S^{\scalebox{0.6}{\text{LV}}}_t)S^{\scalebox{0.6}{\text{LV}}}_{t}dt + \sigma(t,S^{\scalebox{0.6}{\text{LV}}}_t)S^{\scalebox{0.6}{\text{LV}}}_{t}\hspace{1pt}dW^{s}_{t},\hspace{.75em} S^{\scalebox{0.6}{\text{LV}}}_{0}>0,
\end{equation}
when the drift and diffusion coefficients (i.e., $\mu(t,x)x$ and $\sigma(t,x)x$) satisfy a linear growth condition, are 1/2-H\"older continuous in time and Lipschitz continuous in spot. Hence, Corollary \ref{Cor2.8} extends the strong order 1/2 convergence of numerical simulations for LV models to allow dependence on the running maximum under somewhat different model assumptions.

We have thus shown that the Euler discretization of the spot process in \eqref{eq2.1.1} attains the optimal strong convergence order of 1/2 up to a logarithmic factor that is characteristic of approximations of SDEs with globally Lipschitz coefficients \cite{Hofmann:2001,Muller-Gronbach:2002}. As a consequence, the Euler discretization of the spot process also converges with weak order 1/2 (up to a logarithmic factor), which is optimal because the Euler scheme for the running maximum converges with weak order of at most 1/2 (see \cite{Asmussen:1995,Glasserman:2004}) instead of the weak order 1 typical for SDEs with smooth coefficients.

\section{The squared volatility process}\label{sec:vol}

\subsection{The Cox--Ingersoll--Ross process}\label{subsec:CIR}

For the convergence analysis, we need to control both the polynomial and the exponential moments of the CIR process.
\begin{lemma}\label{Lem3.1}
The CIR process $v$ from \eqref{eq2.1.1}:
\begin{enumerate}[(1)]
\item{has bounded moments, i.e.,
\begin{equation}\label{eq3.0.2}
\sup_{t\in[0,T]}\E\big[v_{t}^{p}\big] < \infty,\hspace{1em} \forall\hspace{.5pt}p>-\nu;
\end{equation}}
\item{has uniformly bounded moments, i.e.,
\begin{equation}\label{eq3.0.3}
\E\bigg[\sup_{t\in[0,T]}v_{t}^{p}\bigg] < \infty,\hspace{1em} \forall\hspace{.5pt}p\geq1.
\end{equation}}
\end{enumerate}
\end{lemma}
\begin{proof}
The first part follows from \cite{Dereich:2012} or Theorem 3.1 in \cite{Hurd:2008} whereas the second part follows from Proposition 3.7 in \cite{Cozma:2016a} or Lemma 3.2 in \cite{Dereich:2012}.
\end{proof}

\begin{lemma}\label{Lem3.2}
Let $\lambda>0$.
\begin{enumerate}[(1)]
\item{If
\begin{equation}\label{eq3.0.4}
T < \frac{2}{\sqrt{\smash[b]{(2\lambda\xi^{2}-k^{2})}^{+}}}\Bigg[\frac{\pi}{2}+\arctan\Bigg(\frac{k}{\sqrt{\smash[b]{(2\lambda\xi^{2}-k^{2})^{+}}}}\Bigg)\Bigg],
\end{equation}
then
\begin{equation}\label{eq3.0.5}
\E\bigg[\exp\bigg\{\lambda\int_{0}^{T}{v_{t}\hspace{1pt}dt}\bigg\}\bigg] < \infty.
\end{equation}}
\item{If
\begin{equation}\label{eq3.0.6}
\lambda \leq \frac{1}{8}\hspace{1pt}\xi^{2}(\nu-1)^{2},
\end{equation}
then
\begin{equation}\label{eq3.0.7}
\E\bigg[\exp\bigg\{\lambda\int_{0}^{T}{v_{t}^{-1}\hspace{1pt}dt}\bigg\}\bigg] < \infty.
\end{equation}}
\end{enumerate}
\end{lemma}
\begin{proof}
The first part follows from Proposition 3.1 in \cite{Andersen:2007} or Proposition 3.5 in \cite{Cozma:2016a} and the second part follows from Lemma A.2 in \cite{Bossy:2007} or Theorem 3.1 in \cite{Hurd:2008}.
\end{proof}

\subsection{The full truncation Euler scheme}\label{subsec:FTE}

Throughout this subsection, $\tilde{v}$ and $\bar{v}$ are the processes defined in \eqref{eq2.2.3} and \eqref{eq2.2.4}. First, we include some auxiliary results on the polynomial and exponential integrability of the FTE approximation.

\begin{lemma}\label{Lem3.3}
The FTE scheme has uniformly bounded moments, i.e.,
\begin{equation}\label{eq3.1.1}
\sup_{N\geq1}\hspace{1.5pt}\E\bigg[\sup_{t\in[0,T]}|\tilde{v}_{t}|^p\bigg] < \infty,\hspace{1em} \forall\hspace{.5pt}p\geq1.
\end{equation}
\end{lemma}
\begin{proof}
Follows from a simple application of the Burkholder--Davis--Gundy (BDG) inequality and Proposition 3.7 in \cite{Cozma:2016a}.
\end{proof}

The following lemma, which was proved in \cite{Cozma:2016a}, is concerned with the exponential integrability of the FTE approximation, which is an important ingredient for proving the finiteness of higher moments and the strong convergence of the approximation process $\bar{S}$ defined in \eqref{eq2.2.13}.

\begin{lemma}[Theorem 3.6 in \cite{Cozma:2016a}]\label{Lem3.4}
Let $\lambda>0$ and $N_{0}=\floor{kT}$. If $\lambda<\tfrac{2k^{2}}{\xi^{2}}$ and
\begin{equation}\label{eq3.1.2}
T \leq \frac{2k}{\lambda\xi^{2}}\hspace{1pt},
\end{equation}
or otherwise if $\lambda\geq\tfrac{2k^{2}}{\xi^{2}}$ and
\begin{equation}\label{eq3.1.3}
T \leq \frac{1}{\sqrt{2\lambda}\hspace{1pt}\xi-k}\hspace{1pt},
\end{equation}
then
\begin{equation}\label{eq3.1.4}
\sup_{N>N_{0}}\hspace{1.5pt}\E\bigg[\exp\bigg\{\lambda\int_{0}^{T}{\bar{v}_{t}\hspace{1.5pt}dt}\bigg\}\bigg] < \infty.
\end{equation}
\end{lemma}

Before we can establish the convergence of the approximation process $\bar{S}$, we need the strong convergence of the discretized squared volatility process.

\begin{proposition}[Theorem 1.1 in \cite{Cozma:2017c}]\label{Prop3.5}
Suppose that $\nu>3$ and let $2\leq p<\nu-1$. Then the FTE scheme converges strongly in $L^{p}$ with order 1/2, i.e., there exists a constant $C$ such that, for all $N\geq1$,
\begin{equation}\label{eq3.1.57}
\sup_{t\in\left[0,T\right]}\E\big[|v_{t}-\bar{v}_{t}|^{p}\big]^{\frac{1}{p}} \leq CN^{-\frac{1}{2}}.
\end{equation}
\end{proposition}

\subsection{The backward Euler--Maruyama scheme}\label{subsec:BEM}

Throughout this subsection, $\bar{y}$ and $\bar{v}$ are the processes defined in \eqref{eq2.2.8}. The following lemma is concerned with the finiteness of moments of the BEM approximation.

\begin{lemma}\label{Lem3.6}
Suppose that $\nu\geq1$. Then the BEM scheme has uniformly bounded moments, i.e.,
\begin{equation}\label{eq3.2.1}
\sup_{N\geq1}\hspace{1.5pt}\E\bigg[\sup_{t\in[0,T]}\bar{v}_{t}^p\bigg] < \infty,\hspace{1em} \forall\hspace{.5pt}p\geq1.
\end{equation}
\end{lemma}
\begin{proof}
Follows from a simple application of Lemma 2.5 in \cite{Neuenkirch:2014} to the CIR process.
\end{proof}

The next lemma is concerned with the exponential integrability of the BEM approximation and is a corollary of Proposition 3.4 in \cite{Cozma:2016b}.

\begin{lemma}\label{Lem3.7}
Suppose that $\nu\geq1$ and let $\lambda>0$. If
\begin{equation}\label{eq3.2.2}
T<\frac{1}{\sqrt{2\lambda}\hspace{1pt}\xi}\hspace{1pt},
\end{equation}
then there exists $N_{0}\in\mathbb{N}$ such that
\begin{equation}\label{eq3.2.3}
\sup_{N>N_{0}}\hspace{1.5pt}\E\bigg[\exp\bigg\{\lambda\int_{0}^{T}{\bar{v}_{t}\hspace{1pt}dt}\bigg\}\bigg] < \infty.
\end{equation}
\end{lemma}

Before we can establish the convergence of the approximation process $\bar{S}$, we need the strong convergence of the discretized squared volatility process. For convenience of notation, define
\begin{equation}\label{eq4.0.0}
\bar{t}=\delta t\floor*{\frac{t}{\delta t}}
\end{equation}
for all $t\in[0,T]$.

\begin{proposition}\label{Prop3.8}
Suppose that $\nu>1$ and let $1\leq p<\nu$. Then the BEM scheme converges strongly in $L^{p}$ with order 1/2, i.e., there exists a constant $C$ such that, for all $N\geq1$,
\begin{equation}\label{eq3.2.4}
\sup_{t\in\left[0,T\right]}\E\big[|y_{t}-\bar{y}_{t}|^{p}\big]^{\frac{1}{p}} \leq CN^{-\frac{1}{2}}
\end{equation}
and
\begin{equation}\label{eq3.2.5}
\sup_{t\in\left[0,T\right]}\E\big[|v_{t}-\bar{v}_{t}|^{p}\big]^{\frac{1}{p}} \leq CN^{-\frac{1}{2}}.
\end{equation}
\end{proposition}
\begin{proof}
The triangle inequality yields
\begin{equation}\label{eq3.2.6}
\sup_{t\in\left[0,T\right]}\E\big[|y_{t}-\bar{y}_{t}|^{p}\big] \leq 2^{p-1}\sup_{t\in\left[0,T\right]}\E\big[|y_{t}-y_{\bar{t}}|^{p}\big] + 2^{p-1}\sup_{t\in\left[0,T\right]}\E\big[|y_{\bar{t}}-\tilde{y}_{\bar{t}}|^{p}\big],
\end{equation}
and the bound in \eqref{eq3.2.4} is a direct consequence of Lemma 3.2 and Proposition 3.3 in \cite{Dereich:2012}. Since
\begin{equation}\label{eq3.2.7}
|v_{t}-\bar{v}_{t}| = (y_{t}+\bar{y}_{t})|y_{t}-\bar{y}_{t}|,
\end{equation}
choosing any $p<q<\nu$ and applying H\"older's inequality leads to
\begin{align}\label{eq3.2.8}
\sup_{t\in\left[0,T\right]}\E\big[|v_{t}-\bar{v}_{t}|^{p}\big] &\leq 2^{p-1}\bigg\{\sup_{t\in[0,T]}\E\Big[v_{t}^{\frac{pq}{2(q-p)}}\Big]^{1-\frac{p}{q}} + \sup_{N>N_{0}}\hspace{1.5pt}\sup_{t\in[0,T]}\E\Big[\bar{v}_{t}^{\frac{pq}{2(q-p)}}\Big]^{1-\frac{p}{q}}\bigg\} \nonumber\\[2pt]
&\times \sup_{t\in\left[0,T\right]}\E\big[|y_{t}-\bar{y}_{t}|^{q}\big]^{\frac{p}{q}}.
\end{align}
The bound in \eqref{eq3.2.5} follows from Lemmas \ref{Lem3.1} and \ref{Lem3.6} and \eqref{eq3.2.4}.
\end{proof}

\section{The spot process}\label{sec:spot}

\subsection{The log-spot process}\label{subsec:log-spot}

The following auxiliary result provides upper bounds on the discretization errors in the drift and diffusion functions $\mu$ and $\sigma$.

\begin{lemma}\label{Lem4.3}
Under Assumption \ref{Asm2.2} we have that, for all $u\in[0,T]$,
\begin{equation}\label{eq4.0.1}
\left|\mu\big(u,S_{u},M_{u}\big)-\mu\big(u,\bar{S}_{\bar{u}},\bar{M}_{\bar{u}}\big)\right| \leq \left(C_{\mu,x}+2C_{\mu,m}\right)\sup_{t\in[0,u]}\left|x_{t}-x_{\bar{t}}\right|
+ \left(C_{\mu,x}+C_{\mu,m}\right)\sup_{t\in[0,u]}\left|x_{t}-\bar{x}_{t}\right|
\end{equation}
and
\begin{align}\label{eq4.0.2}
\left|\sigma\big(u,S_{u},M_{u}\big)-\sigma\big(\bar{u},\bar{S}_{\bar{u}},\bar{M}_{\bar{u}}\big)\right| &\leq C_{\sigma,t}\sqrt{\delta t} + \left(C_{\sigma,x}+2C_{\sigma,m}\right)\sup_{t\in[0,u]}\left|x_{t}-x_{\bar{t}}\right| \nonumber\\[2pt]
&+ \left(C_{\sigma,x}+C_{\sigma,m}\right)\sup_{t\in[0,u]}\left|x_{t}-\bar{x}_{t}\right|.
\end{align}
\end{lemma}
\begin{proof}
See Appendix \ref{sec:aux2}.
\end{proof}

Since the choice of discretization scheme for the squared volatility process makes little difference in the subsequent proofs, we henceforth denote by $\bar{v}$ both the FTE and the BEM discretizations. For the convergence analysis, we need to control the polynomial moments of the log-spot process and its approximation.

\begin{lemma}\label{Lem4.4}
The following statements hold under Assumption \ref{Asm2.1}.
\begin{enumerate}[(1)]
\item{The log-spot process has uniformly bounded moments, i.e.,
\begin{equation}\label{eq4.1.30}
\E\bigg[\sup_{t\in[0,T]}|x_{t}|^{p}\bigg] < \infty,\hspace{1em} \forall\hspace{.5pt}p\geq1.
\end{equation}}
\item{The approximated log-spot process has uniformly bounded moments, i.e.,
\begin{equation}\label{eq4.1.31}
\sup_{N\geq1}\hspace{1.5pt}\E\bigg[\sup_{t\in[0,T]}|\bar{x}_{t}|^{p}\bigg] < \infty,\hspace{1em} \forall\hspace{.5pt}p\geq1.
\end{equation}}
\end{enumerate}
\end{lemma}
\begin{proof}
(1) Note from Jensen's inequality that it suffices to consider $p\geq2$. Recall from \eqref{eq2.1.1} that
\begin{equation}\label{eq4.1.32}
x_{t} = x_{0} + \int_{0}^{t}{\mu\big(u,S_{u},M_{u}\big)du} - \frac{1}{2}\int_{0}^{t}{\sigma^{2}\big(u,S_{u},M_{u}\big)v_{u}\,du} + \int_{0}^{t}{\sigma\big(u,S_{u},M_{u}\big)\sqrt{v_{u}}\,dW^{s}_{u}}.
\end{equation}
Using the H\"older and BDG inequalities and Fubini's theorem, we deduce that
\begin{align}\label{eq4.1.33}
\E\bigg[\sup_{t\in[0,T]}|x_{t}|^{p}\bigg] &\leq 4^{p-1}\left(|x_{0}|^{p}+\mu_{max}^{p}T^{p}\right) + 2^{p-2}\sigma_{max}^{2p}T^{p}\sup_{t\in[0,T]}\E\big[v_{t}^{p}\big] \nonumber\\[2pt]
&+ 4^{p-1}\sigma_{max}^{p}T^{\frac{p}{2}}C\sup_{t\in[0,T]}\E\big[v_{t}^{p/2}\big],
\end{align}
for some non-negative constant $C$, and the right-hand side is finite by Lemma \ref{Lem3.1}.

\noindent
(2) Recall from \eqref{eq2.2.12} that
\begin{equation}\label{eq4.1.34}
\bar{x}_{t} = x_{0} + \int_{0}^{t}{\bar{\mu}\big(u,\bar{S}_{u},\bar{M}_{u}\big)du} - \frac{1}{2}\int_{0}^{t}{\bar{\sigma}^{2}\big(u,\bar{S}_{u},\bar{M}_{u}\big)\bar{v}_{u}\,du} + \int_{0}^{t}{\bar{\sigma}\big(u,\bar{S}_{u},\bar{M}_{u}\big)\sqrt{\bar{v}_{u}}\,dW^{s}_{u}}.
\end{equation}
Proceeding as before, we deduce that
\begin{align}\label{eq4.1.35}
\sup_{N\geq1}\hspace{1.5pt}\E\bigg[\sup_{t\in[0,T]}|\bar{x}_{t}|^{p}\bigg] &\leq 4^{p-1}\left(|x_{0}|^{p}+\mu_{max}^{p}T^{p}\right) + 2^{p-2}\sigma_{max}^{2p}T^{p}\sup_{N\geq1}\hspace{1.5pt}\sup_{t\in[0,T]}\E\big[\bar{v}_{t}^{p}\big] \nonumber\\[2pt]
&+ 4^{p-1}\sigma_{max}^{p}T^{\frac{p}{2}}C\sup_{N\geq1}\hspace{1.5pt}\sup_{t\in[0,T]}\E\big[\bar{v}_{t}^{p/2}\big],
\end{align}
and the conclusion follows from Lemmas \ref{Lem3.3} and \ref{Lem3.6}.
\end{proof}

The following result is concerned with the uniform convergence in $L^{p}$ with order 1/2 (up to a logarithmic factor) of the approximated log-spot process. In the special case of constant drift and diffusion functions $\mu$ and $\sigma$, the SPDV model \eqref{eq2.1.1} collapses to the Heston stochastic volatility model and we notice from \eqref{eq2.3.3} -- \eqref{eq2.3.5} that $T^{*}_{x}(p)=\infty$, for all $1\leq p<p^{*}(\nu)$. For the LE--BEM scheme, i.e., when the LE and the BEM schemes are employed in the discretization of the spot process and its squared volatility, respectively, this result was proved in Corollary 5.5 in \cite{Kloeden:2012}. Furthermore, the extension to the LE--FTE scheme is straightforward. However, the simple argument employed to prove Proposition \ref{Prop4.5} under a purely stochastic volatility model does not apply to the general case of non-trivial drift and diffusion functions $\mu$ and $\sigma$, even without path-dependency. In this case, we require more advanced techniques in order to overcome the technical challenges. We also mention that in the case of no stochastic volatility (e.g., take $v_{0}=\theta=1$ and $\xi=0$), the SPDV model \eqref{eq2.1.1} collapses to a path-dependent volatility model and we notice from \eqref{eq2.3.1}, \eqref{eq2.3.2b} and \eqref{eq2.3.3} -- \eqref{eq2.3.5} that $\nu=p^{*}(\nu)=\infty$ and $T^{*}_{x}(p)=\infty$, for all $p\geq1$. In this case, the analysis involved in Proposition \ref{Prop4.5} becomes somewhat simpler, as will be clear from the proof.

\begin{proposition}\label{Prop4.5}
Suppose that Assumptions \ref{Asm2.1} and \ref{Asm2.2} hold and that $\nu>\nu^{*}$, with $\nu^{*}$ defined in \eqref{eq2.3.2a}. Then for all $2\leq p<p^{*}(\nu)$ and $T<T^{*}_{x}(p)$, with $p^{*}$ defined in \eqref{eq2.3.2b} and $T^{*}_{x}$ given in \eqref{eq2.3.3}, there exists a constant $C$ such that, for all $N\geq1$,
\begin{equation}\label{eq4.1.36}
\E\bigg[\sup_{t\in[0,T]}|x_{t}-\bar{x}_{t}|^{p}\bigg]^{\frac{1}{p}} \leq C\sqrt{\frac{\log(2N)}{N}}\hspace{1pt}.
\end{equation}
\end{proposition}
\begin{proof}
First, by a continuity argument, we can find $p<q<p^{*}(\nu)$ such that
\begin{equation}\label{eq4.1.37.4}
T < \frac{2}{\sqrt{\smash[b]{(\tilde{\varphi}(p,q)-k^{2})^{+}}}}\Bigg[\frac{\pi}{2} + \arctan\Bigg(\frac{k}{\sqrt{\smash[b]{(\tilde{\varphi}(p,q)-k^{2})^{+}}}}\Bigg)\Bigg].
\end{equation}
For convenience of notation, define
\begin{equation}\label{eq4.1.38.1}
e^{x}_{t}=x_{t}-\bar{x}_{t},\hspace{.75em} e^{x}_{0}=0,
\end{equation}
and
\begin{equation}\label{eq4.1.38.2}
\Delta x_{t} = x_{t}-x_{\bar{t}}.
\end{equation}
Let $\tau$ be a stopping time. Applying It\^o's formula to the $\mathcal{C}^{2}$ function $f(e^{x}_{t\wedge\tau})=|e^{x}_{t\wedge\tau}|^{q}$ yields
\begin{align}\label{eq4.1.39}
|e^{x}_{t\wedge\tau}|^{q} &= q\int_{0}^{t\wedge\tau}{|e^{x}_{u}|^{q-1}\sgn(e^{x}_{u})\,de^{x}_{u}} + \frac{1}{2}\hspace{1pt}q(q-1)\int_{0}^{t\wedge\tau}{|e^{x}_{u}|^{q-2}\,d\langle e^{x}\rangle_{u}},
\end{align}
where $\sgn(e^{x})=1$ if $e^{x}>0$ and $\sgn(e^{x})=-1$ otherwise, and hence
\begin{align}\label{eq4.1.40}
|e^{x}_{t\wedge\tau}|^{q} &= q\int_{0}^{t\wedge\tau}{|e^{x}_{u}|^{q-1}\sgn(e^{x}_{u})\left(\mu\big(u,S_{u},M_{u}\big)-\mu\big(u,\bar{S}_{\bar{u}},\bar{M}_{\bar{u}}\big)\right)du} \nonumber\\[2pt]
&- \frac{1}{2}\hspace{1pt}q\int_{0}^{t\wedge\tau}{|e^{x}_{u}|^{q-1}\sgn(e^{x}_{u})\left(v_{u}\sigma^{2}\big(u,S_{u},M_{u}\big)-\bar{v}_{u}\sigma^{2}\big(\bar{u},\bar{S}_{\bar{u}},\bar{M}_{\bar{u}}\big)\right)du} \nonumber\\[2pt]
&+ q\int_{0}^{t\wedge\tau}{|e^{x}_{u}|^{q-1}\sgn(e^{x}_{u})\left(\sqrt{v_{u}}\hspace{1pt}\sigma\big(u,S_{u},M_{u}\big)-\sqrt{\bar{v}_{u}}\hspace{1pt}\sigma\big(\bar{u},\bar{S}_{\bar{u}},\bar{M}_{\bar{u}}\big)\right)dW^{s}_{u}} \nonumber\\[2pt]
&+ \frac{1}{2}\hspace{1pt}q(q-1)\int_{0}^{t\wedge\tau}{|e^{x}_{u}|^{q-2}\left(\sqrt{v_{u}}\hspace{1pt}\sigma\big(u,S_{u},M_{u}\big)-\sqrt{\bar{v}_{u}}\hspace{1pt}\sigma\big(\bar{u},\bar{S}_{\bar{u}},\bar{M}_{\bar{u}}\big)\right)^{2}du}.
\end{align}
Taking the supremum over $[0,T]$ and then expectations on both sides, we deduce that
\begin{align}\label{eq4.1.41}
\E\!\bigg[\sup_{t\in[0,T]}|e^{x}_{t\wedge\tau}|^{q}\bigg] &\leq q\E\bigg[\int_{0}^{T\wedge\hspace{.5pt}\tau}{|e^{x}_{u}|^{q-1}\left|\mu\big(u,S_{u},M_{u}\big)-\mu\big(u,\bar{S}_{\bar{u}},\bar{M}_{\bar{u}}\big)\right|du}\bigg] \nonumber\\[2pt]
&+ \frac{1}{2}\hspace{1pt}q\E\bigg[\int_{0}^{T\wedge\hspace{.5pt}\tau}{|e^{x}_{u}|^{q-1}\left|v_{u}\sigma^{2}\big(u,S_{u},M_{u}\big)-\bar{v}_{u}\sigma^{2}\big(\bar{u},\bar{S}_{\bar{u}},\bar{M}_{\bar{u}}\big)\right|du}\bigg] \nonumber\\[2pt]
&+ q\E\bigg[\sup_{t\in[0,T]}\int_{0}^{t\wedge\tau}{|e^{x}_{u}|^{q-1}\sgn(e^{x}_{u})\left(\sqrt{v_{u}}\hspace{1pt}\sigma\big(u,S_{u},M_{u}\big)-\sqrt{\bar{v}_{u}}\hspace{1pt}\sigma\big(\bar{u},\bar{S}_{\bar{u}},\bar{M}_{\bar{u}}\big)\right)dW^{s}_{u}}\bigg] \nonumber\\[2pt]
&+ \frac{1}{2}\hspace{1pt}q(q\hspace{-1pt}-\hspace{-1pt}1)\E\bigg[\int_{0}^{T\wedge\hspace{.5pt}\tau}{|e^{x}_{u}|^{q-2}\left|\sqrt{v_{u}}\hspace{1pt}\sigma\big(u,S_{u},M_{u}\big)-\sqrt{\bar{v}_{u}}\hspace{1pt}\sigma\big(\bar{u},\bar{S}_{\bar{u}},\bar{M}_{\bar{u}}\big)\right|^{2}du}\bigg].
\end{align}
We can show that the stochastic integral in \eqref{eq4.1.40} is a true martingale by a simple application of H\"older's inequality and Lemmas \ref{Lem3.1}, \ref{Lem3.3}, \ref{Lem3.6} and \ref{Lem4.4}.

Let $\lambda\in(0,1)$. Using a sharp maximal inequality for continuous-path martingales starting at zero (Corollary 4.4 in \cite{Osekowski:2010}) and the arithmetic mean-geometric mean (AM-GM) inequality yields
\begin{align}\label{eq4.1.43}
&\E\bigg[\sup_{t\in[0,T]}\int_{0}^{t\wedge\tau}{|e^{x}_{u}|^{q-1}\sgn(e^{x}_{u})\left(\sqrt{v_{u}}\hspace{1pt}\sigma\big(u,S_{u},M_{u}\big)-\sqrt{\bar{v}_{u}}\hspace{1pt}\sigma\big(\bar{u},\bar{S}_{\bar{u}},\bar{M}_{\bar{u}}\big)\right)dW^{s}_{u}}\bigg] \nonumber\\[2pt]
\leq\hspace{3pt} &\beta_{0}\E\Bigg[\bigg(\int_{0}^{T\wedge\hspace{.5pt}\tau}{|e^{x}_{u}|^{2(q-1)}\left|\sqrt{v_{u}}\hspace{1pt}\sigma\big(u,S_{u},M_{u}\big)-\sqrt{\bar{v}_{u}}\hspace{1pt}\sigma\big(\bar{u},\bar{S}_{\bar{u}},\bar{M}_{\bar{u}}\big)\right|^{2}du}\bigg)^{\frac{1}{2}}\Bigg] \nonumber\\[2pt]
\leq\hspace{3pt} &\beta_{0}\E\Bigg[\bigg(\sup_{t\in[0,T]}|e^{x}_{t\wedge\tau}|^{q}\int_{0}^{T\wedge\hspace{.5pt}\tau}{|e^{x}_{u}|^{q-2}\left|\sqrt{v_{u}}\hspace{1pt}\sigma\big(u,S_{u},M_{u}\big)-\sqrt{\bar{v}_{u}}\hspace{1pt}\sigma\big(\bar{u},\bar{S}_{\bar{u}},\bar{M}_{\bar{u}}\big)\right|^{2}du}\bigg)^{\frac{1}{2}}\Bigg] \nonumber\\[2pt]
\leq\hspace{3pt} &\frac{\lambda}{q}\E\bigg[\sup_{t\in[0,T]}|e^{x}_{t\wedge\tau}|^{q}\bigg] + \frac{q\beta_{0}^{2}}{4\lambda}\E\bigg[\int_{0}^{T\wedge\hspace{.5pt}\tau}{|e^{x}_{u}|^{q-2}\left|\sqrt{v_{u}}\hspace{1pt}\sigma\big(u,S_{u},M_{u}\big)-\sqrt{\bar{v}_{u}}\hspace{1pt}\sigma\big(\bar{u},\bar{S}_{\bar{u}},\bar{M}_{\bar{u}}\big)\right|^{2}du}\bigg],
\end{align}
with $\beta_{0}$ defined in \eqref{eq2.3.2c}. Substituting back into \eqref{eq4.1.41} with \eqref{eq4.1.43} and after some rearrangements, we deduce that
\begin{align}\label{eq4.1.44}
\E\bigg[\sup_{t\in[0,T]}|e^{x}_{t\wedge\tau}|^{q}\bigg] &\leq \frac{q}{1-\lambda}\E\bigg[\int_{0}^{T\wedge\hspace{.5pt}\tau}{|e^{x}_{u}|^{q-1}\left|\mu\big(u,S_{u},M_{u}\big)-\mu\big(u,\bar{S}_{\bar{u}},\bar{M}_{\bar{u}}\big)\right|du}\bigg] \nonumber\\[2pt]
&+ \frac{q}{2(1-\lambda)}\E\bigg[\int_{0}^{T\wedge\hspace{.5pt}\tau}{|e^{x}_{u}|^{q-1}\left|v_{u}\sigma^{2}\big(u,S_{u},M_{u}\big)-\bar{v}_{u}\sigma^{2}\big(\bar{u},\bar{S}_{\bar{u}},\bar{M}_{\bar{u}}\big)\right|du}\bigg] \nonumber\\[2pt]
&+ qc_{q}(\lambda)\E\bigg[\int_{0}^{T\wedge\hspace{.5pt}\tau}{|e^{x}_{u}|^{q-2}\left|\sqrt{v_{u}}\hspace{1pt}\sigma\big(u,S_{u},M_{u}\big)-\sqrt{\bar{v}_{u}}\hspace{1pt}\sigma\big(\bar{u},\bar{S}_{\bar{u}},\bar{M}_{\bar{u}}\big)\right|^{2}du}\bigg],
\end{align}
where we defined, for brevity,
\begin{equation}\label{eq4.1.45}
c_{q}(\lambda) = \frac{q-1}{2(1-\lambda)}+\frac{q\beta_{0}^{2}}{4\lambda(1-\lambda)}\hspace{1pt}.
\end{equation}
For any $n\in\mathbb{N}$ and $z_{k}\geq0$, $d_{k}>0$, for all $1\leq k\leq n$, the Cauchy--Schwarz inequality yields
\begin{equation}\label{eq4.1.46}
\Bigg(\sum_{k=1}^{n}{z_{k}}\Bigg)^{2} \leq \Bigg(\sum_{k=1}^{n}{z_{k}^{2}d_{k}}\Bigg)\Bigg(\sum_{k=1}^{n}{d_{k}^{-1}}\Bigg).
\end{equation}
Let $\eta>0$. Using Lemma \ref{Lem4.3} and \eqref{eq4.1.46} with $n=3$, $d_{1}=d_{2}=2\eta^{-1}(1+\eta)$ and $d_{3}=1+\eta$, we get
\begin{align}\label{eq4.1.47}
\left|\sigma\big(u,S_{u},M_{u}\big)-\sigma\big(\bar{u},\bar{S}_{\bar{u}},\bar{M}_{\bar{u}}\big)\right|^{2} &\leq 2\eta^{-1}(1+\eta)C_{\sigma,t}^{2}\delta t + 2\eta^{-1}(1+\eta)\left(C_{\sigma,x}+2C_{\sigma,m}\right)^{2}\sup_{t\in[0,u]}\left|\Delta x_{t}\right|^{2} \nonumber\\[2pt]
&+ (1+\eta)\left(C_{\sigma,x}+C_{\sigma,m}\right)^{2}\sup_{t\in[0,u]}|e^{x}_{t}|^{2}.
\end{align}
Next, using Lemma \ref{Lem4.3}, \eqref{eq4.1.46} with $n=2$, $d_{1}=\eta^{-1}(1+\eta)$ and $d_{2}=1+\eta$, as well as \eqref{eq4.1.47}, after some rearrangements, we deduce from \eqref{eq4.1.44} that
\begin{align}\label{eq4.1.48}
\E\bigg[\sup_{t\in[0,T]}|e^{x}_{t\wedge\tau}|^{q}\bigg] &\leq \frac{q}{1-\lambda}\left(C_{\mu,x}+C_{\mu,m}\right)\E\bigg[\int_{0}^{T\wedge\hspace{.5pt}\tau}{\sup_{t\in[0,u]}|e^{x}_{t}|^{q}\,du}\bigg] \nonumber\\[2pt]
&+ \frac{q}{1-\lambda}\hspace{1pt}\sigma_{max}\left(C_{\sigma,x}+C_{\sigma,m}\right)\E\bigg[\int_{0}^{T\wedge\hspace{.5pt}\tau}{v_{u}\sup_{t\in[0,u]}|e^{x}_{t}|^{q}\,du}\bigg] \nonumber\\[2pt]
&+ qc_{q}(\lambda)(1+\eta)^{2}\left(C_{\sigma,x}+C_{\sigma,m}\right)^{2}\E\bigg[\int_{0}^{T\wedge\hspace{.5pt}\tau}{v_{u}\sup_{t\in[0,u]}|e^{x}_{t}|^{q}\,du}\bigg] \nonumber\\[2pt]
&+ \frac{q}{1-\lambda}\hspace{1pt}\sigma_{max}C_{\sigma,t}T^{\frac{1}{2}}\hspace{1pt}\E\bigg[\int_{0}^{T\wedge\hspace{.5pt}\tau}{v_{u}\sup_{t\in[0,u]}|e^{x}_{t}|^{q-1}N^{-\frac{1}{2}}\,du}\bigg] \nonumber\\[2pt]
&+ \frac{q}{1-\lambda}\left(C_{\mu,x}+2C_{\mu,m}\right)\E\bigg[\int_{0}^{T\wedge\hspace{.5pt}\tau}{\sup_{t\in[0,u]}|e^{x}_{t}|^{q-1}\sup_{t\in[0,u]}\left|\Delta x_{t}\right|du}\bigg] \nonumber\\[2pt]
&+ \frac{q}{1-\lambda}\hspace{1pt}\sigma_{max}\left(C_{\sigma,x}+2C_{\sigma,m}\right)\E\bigg[\int_{0}^{T\wedge\hspace{.5pt}\tau}{v_{u}\sup_{t\in[0,u]}|e^{x}_{t}|^{q-1}\sup_{t\in[0,u]}\left|\Delta x_{t}\right|du}\bigg] \nonumber\\[2pt]
&+ \frac{q}{2(1-\lambda)}\hspace{1pt}\sigma_{max}^{2}\E\bigg[\int_{0}^{T\wedge\hspace{.5pt}\tau}{\sup_{t\in[0,u]}|e^{x}_{t}|^{q-1}\left|v_{u}-\bar{v}_{u}\right|du}\bigg] \nonumber\\[2pt]
&+ 2qc_{q}(\lambda)\eta^{-1}(1+\eta)^{2}C_{\sigma,t}^{2}T\E\bigg[\int_{0}^{T\wedge\hspace{.5pt}\tau}{v_{u}\sup_{t\in[0,u]}|e^{x}_{t}|^{q-2}N^{-1}\,du}\bigg] \nonumber\\[2pt]
&+ 2qc_{q}(\lambda)\eta^{-1}(1+\eta)^{2}\left(C_{\sigma,x}+2C_{\sigma,m}\right)^{2}\E\bigg[\int_{0}^{T\wedge\hspace{.5pt}\tau}{v_{u}\sup_{t\in[0,u]}|e^{x}_{t}|^{q-2}\sup_{t\in[0,u]}\left|\Delta x_{t}\right|^{2}du}\bigg] \nonumber\\[2pt]
&+ (1-\gamma^{*})qc_{q}(\lambda)\eta^{-1}(1+\eta)\sigma_{max}^{2}\E\bigg[\int_{0}^{T\wedge\hspace{.5pt}\tau}{\sup_{t\in[0,u]}|e^{x}_{t}|^{q-2}\left|\sqrt{v_{u}}-\sqrt{\bar{v}_{u}}\hspace{1pt}\right|^{2}du}\bigg] \nonumber\\[2pt]
&+ \gamma^{*}qc_{q}(\lambda)\eta^{-1}(1+\eta)\sigma_{max}^{2}\E\bigg[\int_{0}^{T\wedge\hspace{.5pt}\tau}{\sup_{t\in[0,u]}|e^{x}_{t}|^{q-2}\left|\sqrt{v_{u}}-\sqrt{\bar{v}_{u}}\hspace{1pt}\right|^{2}du}\bigg],
\end{align}
where $\gamma^{\scalebox{0.6}{\text{FTE}}}=0$ and $\gamma^{\scalebox{0.6}{\text{BEM}}}=1$. Since $\nu\geq1$, the process $v$ has almost surely strictly positive paths and we can bound the term before the last on the right-hand side of \eqref{eq4.1.48} from above by using
\begin{equation}\label{eq4.1.49}
\left|\sqrt{v_{u}}-\sqrt{\bar{v}_{u}}\hspace{1pt}\right|^{2} \leq v_{u}^{-1}\left|v_{u}-\bar{v}_{u}\right|^{2}.
\end{equation}
For any $a,b\geq0$ and $j\in\{1,2\}$, Young's inequality yields
\begin{equation}\label{eq4.1.50}
a^{q-j}b^{j} = \Big(\eta^{\frac{j(q-j)}{q}}a^{q-j}\Big)\Big(\eta^{-\frac{j(q-j)}{q}}b^{j}\Big) \leq \frac{q-j}{q}\hspace{1pt}\eta^{j}a^{q} + \frac{j}{q}\hspace{1pt}\eta^{j-q}b^{q}.
\end{equation}
Going back to \eqref{eq4.1.48}, using \eqref{eq4.1.49}, \eqref{eq4.1.50} (with $\eta^{\frac{3}{2}}$ instead of $\eta$ for the term before the last) and Fubini's theorem leads to an upper bound
\begin{align}\label{eq4.1.52}
\E\bigg[\sup_{t\in[0,T]}|e^{x}_{t\wedge\tau}|^{q}\bigg] &\leq \bigg\{\frac{\eta^{1-q}}{1-\lambda}\big[4(1-\lambda)c_{q}(\lambda)(1+\eta)^{2}C_{\sigma,t}^{2}T^{2}+\sigma_{max}C_{\sigma,t}T^{\frac{3}{2}}\big]\sup_{t\in[0,T]}\E\left[v_{t}\right]N^{-\frac{q}{2}} \nonumber\\[4pt]
&+ 2(1-\gamma^{*})\eta^{2-\frac{3}{2}q}(1+\eta)c_{q}(\lambda)\sigma_{max}^{2}T\sup_{t\in[0,T]}\E\big[v_{t}^{-1}|v_{t}-\bar{v}_{t}|^{q}\big] \nonumber\\[6pt]
&+ 2\gamma^{*}\eta^{1-q}(1+\eta)c_{q}(\lambda)\sigma_{max}^{2}T\sup_{t\in[0,T]}\E\big[|\sqrt{v_{t}}-\sqrt{\bar{v}_{t}}\hspace{1pt}|^{q}\big] \nonumber\\[0pt]
&+ \frac{\eta^{1-q}}{2(1-\lambda)}\hspace{1pt}\sigma_{max}^{2}T\sup_{t\in[0,T]}\E\big[|v_{t}-\bar{v}_{t}|^{q}\big] + \frac{\eta^{1-q}}{1-\lambda}\left(C_{\mu,x}+2C_{\mu,m}\right)T\E\bigg[\sup_{t\in[0,T]}\left|\Delta x_{t}\right|^{q}\bigg] \nonumber\\[0pt]
&+ 4\eta^{1-q}(1+\eta)^{2}c_{q}(\lambda)\left(C_{\sigma,x}+2C_{\sigma,m}\right)^{2}T\E\bigg[\sup_{t\in[0,T]}v_{t}\hspace{1pt}\sup_{t\in[0,T]}\left|\Delta x_{t}\right|^{q}\bigg] \nonumber\\[0pt]
&+ \frac{\eta^{1-q}}{1-\lambda}\hspace{1pt}\sigma_{max}\left(C_{\sigma,x}+2C_{\sigma,m}\right)T\E\bigg[\sup_{t\in[0,T]}v_{t}\hspace{1pt}\sup_{t\in[0,T]}\left|\Delta x_{t}\right|^{q}\bigg]\bigg\} \nonumber\\[1pt]
&+ \E\bigg[\int_{0}^{T\wedge\hspace{.5pt}\tau}\sup_{t\in[0,u]}|e^{x}_{t}|^{q}\hspace{3pt}\bigg\{v_{u}\bigg(\frac{q}{1-\lambda}\hspace{1pt}\sigma_{max}\left(C_{\sigma,x}+C_{\sigma,m}\right) + qc_{q}(\lambda)\left(C_{\sigma,x}+C_{\sigma,m}\right)^{2} \nonumber\\[2pt]
&+ \eta(2+\eta)qc_{q}(\lambda)\left(C_{\sigma,x}+C_{\sigma,m}\right)^{2} + \frac{\eta(q-1)}{1-\lambda}\hspace{1pt}\sigma_{max}\big[C_{\sigma,t}T^{\frac{1}{2}}+C_{\sigma,x}+2C_{\sigma,m}\big] \nonumber\\[5pt]
&+ 2\eta(1+\eta)^{2}(q-2)c_{q}(\lambda)\big[C_{\sigma,t}^{2}T+\left(C_{\sigma,x}+2C_{\sigma,m}\right)^{2}\big]\bigg) + \bigg(\frac{q}{1-\lambda}\left(C_{\mu,x}+C_{\mu,m}\right) \nonumber\\[3pt]
&+ \frac{\eta(q-1)}{1-\lambda}\left(C_{\mu,x}+2C_{\mu,m}\right) + \frac{\eta(q-1)}{2(1-\lambda)}\hspace{1pt}\sigma_{max}^{2} + \gamma^{*}\eta(1+\eta)(q-2)c_{q}(\lambda)\sigma_{max}^{2}\bigg) \nonumber\\[3pt]
&+ v_{u}^{-1}(1-\gamma^{*})\eta^{2}(1+\eta)(q-2)c_{q}(\lambda)\sigma_{max}^{2}\bigg\}\,du\bigg].
\end{align}
We choose $\lambda=\lambda_{q}$ that minimizes the function $f_{q}(\lambda):(0,1)\rightarrow\RR$ given by
\begin{equation}\label{eq4.1.53}
f_{q}(\lambda) = \frac{q}{1-\lambda}\hspace{1pt}\sigma_{max}\left(C_{\sigma,x}+C_{\sigma,m}\right) + q\bigg(\frac{q-1}{2(1-\lambda)}+\frac{q\beta_{0}^{2}}{4\lambda(1-\lambda)}\bigg)\left(C_{\sigma,x}+C_{\sigma,m}\right)^{2}.
\end{equation}
For brevity, define
\begin{equation}\label{eq4.1.54}
\Delta_{\sigma} = \frac{C_{\sigma,x}+C_{\sigma,m}}{\sigma_{max}}\hspace{1pt}.
\end{equation}
Looking at the first derivative of $f_{q}$, we find its unique positive root
\begin{equation}\label{eq4.1.55}
\lambda_{q} = \frac{-q\beta_{0}^{2}\Delta_{\sigma}+\beta_{0}\sqrt{q^{2}\beta_{0}^{2}\Delta_{\sigma}^{2}+2q\Delta_{\sigma}\left(2+(q-1)\Delta_{\sigma}\right)}}{4+2(q-1)\Delta_{\sigma}}\hspace{1pt},
\end{equation}
which clearly satisfies $\lambda_{q}\in(0,1)$. Some straightforward computations lead to
\begin{equation}\label{eq4.1.56}
f_{q}(\lambda_{q}) = \frac{q^{2}\beta_{0}^{2}}{4\lambda_{q}^{2}}\left(C_{\sigma,x}+C_{\sigma,m}\right)^{2} = \frac{1}{4}\hspace{1pt}\sigma_{max}^{2}\Big(\sqrt{q^{2}(2+\beta_{0}^{2})\Delta_{\sigma}^{2}+2q\Delta_{\sigma}\big(2-\Delta_{\sigma}\big)}+q\beta_{0}\Delta_{\sigma}\Big)^{2}.
\end{equation}
Next, we bound the first seven terms on the right-hand side of \eqref{eq4.1.52} from above. On the one hand, for the FTE discretization of the squared volatility process, note that we can find $r>1$ such that
\begin{equation}\label{eq4.1.57}
\frac{\nu}{\nu-1} < \frac{r}{r-1} < \frac{\nu-1}{q}\hspace{1pt}.
\end{equation}
Using H\"older's inequality, Lemma \ref{Lem3.1} and Proposition \ref{Prop3.5}, we deduce that there exists a constant $C$ such that, for all $N$ large enough,
\begin{equation}\label{eq4.1.58}
\sup_{t\in[0,T]}\E\big[v_{t}^{-1}|v_{t}-\bar{v}_{t}|^{q}\big] \leq \sup_{t\in[0,T]}\E\big[v_{t}^{-r}\big]^{\frac{1}{r}}\sup_{t\in[0,T]}\E\Big[|v_{t}-\bar{v}_{t}|^{\frac{rq}{r-1}}\Big]^{\frac{r-1}{r}} \leq CN^{-\frac{q}{2}}
\end{equation}
and
\begin{equation}\label{eq4.1.59}
\sup_{t\in[0,T]}\E\big[|v_{t}-\bar{v}_{t}|^{q}\big] \leq CN^{-\frac{q}{2}}.
\end{equation}
On the other hand, for the BEM discretization of the squared volatility process, we know from Proposition \ref{Prop3.8} that there exists a constant $C$ such that, for all $N$ large enough,
\begin{equation}\label{eq4.1.60}
\sup_{t\in[0,T]}\E\big[|\sqrt{v_{t}}-\sqrt{\bar{v}_{t}}\hspace{1pt}|^{q}\big] \leq CN^{-\frac{q}{2}}
\end{equation}
and
\begin{equation}\label{eq4.1.61}
\sup_{t\in[0,T]}\E\big[|v_{t}-\bar{v}_{t}|^{q}\big] \leq CN^{-\frac{q}{2}}.
\end{equation}
Furthermore, using the Cauchy--Schwarz inequality and Lemma \ref{Lem3.1} and applying Theorem 1 in \cite{Fischer:2009} to the log-spot process from \eqref{eq4.1.32}, we conclude that there exists a constant $C$ such that, for all $N\geq1$,
\begin{equation}\label{eq4.1.62}
\E\bigg[\sup_{t\in[0,T]}v_{t}\hspace{1pt}\sup_{t\in[0,T]}\left|\Delta x_{t}\right|^{q}\bigg] \leq \E\bigg[\sup_{t\in[0,T]}v_{t}^{2}\bigg]^{\frac{1}{2}}\E\bigg[\sup_{t\in[0,T]}\left|\Delta x_{t}\right|^{2q}\bigg]^{\frac{1}{2}} \leq C\bigg(\frac{N}{\log(2N)}\bigg)^{-\frac{q}{2}}
\end{equation}
and
\begin{equation}\label{eq4.1.63}
\E\bigg[\sup_{t\in[0,T]}\left|\Delta x_{t}\right|^{q}\bigg] \leq C\bigg(\frac{N}{\log(2N)}\bigg)^{-\frac{q}{2}}.
\end{equation}
For convenience, define the strictly increasing stochastic process
\begin{equation}\label{eq4.1.64}
g_{q,\eta}(t) = \int_{0}^{t}{\Big(a_{q,\eta} + \big(f_{q}(\lambda_{q})+\eta b_{q,\eta}\big)v_{u} + \eta^{2}c_{q,\eta}v_{u}^{-1}\Big)du},
\end{equation}
where
\begin{align}
a_{q,\eta} &= \frac{q}{1-\lambda_{q}}\left(C_{\mu,x}+C_{\mu,m}\right) + \frac{\eta(q-1)}{1-\lambda_{q}}\left(C_{\mu,x}+2C_{\mu,m}\right) + \frac{\eta(q-1)}{2(1-\lambda_{q})}\hspace{1pt}\sigma_{max}^{2} \nonumber\\[2pt]
&+ \gamma^{*}\eta(1+\eta)(q-2)c_{q}(\lambda_{q})\sigma_{max}^{2}, \label{eq4.1.64.1}\\[4pt]
b_{q,\eta} &= (2+\eta)qc_{q}(\lambda_{q})\left(C_{\sigma,x}+C_{\sigma,m}\right)^{2} + \frac{q-1}{1-\lambda_{q}}\hspace{1pt}\sigma_{max}\big[C_{\sigma,t}T^{\frac{1}{2}}+C_{\sigma,x}+2C_{\sigma,m}\big] \nonumber\\[1pt]
&+ 2(1+\eta)^{2}(q-2)c_{q}(\lambda_{q})\big[C_{\sigma,t}^{2}T+\left(C_{\sigma,x}+2C_{\sigma,m}\right)^{2}\big], \label{eq4.1.64.2}\\[7pt]
c_{q,\eta} &= (1-\gamma^{*})(1+\eta)(q-2)c_{q}(\lambda_{q})\sigma_{max}^{2}. \label{eq4.1.64.3}
\end{align}
Substituting back into \eqref{eq4.1.52} with \eqref{eq4.1.58} -- \eqref{eq4.1.64}, we conclude that there exists a constant $C_{q,\eta}>0$ such that, for all $N$ large enough,
\begin{equation}\label{eq4.1.65}
\E\bigg[\sup_{t\in[0,T]}|e^{x}_{t\wedge\tau}|^{q}\bigg] \leq C_{q,\eta}\bigg(\frac{N}{\log(2N)}\bigg)^{-\frac{q}{2}} + \E\bigg[\int_{0}^{T\wedge\hspace{.5pt}\tau}\sup_{t\in[0,T]}|e^{x}_{t\wedge u}|^{q}\,dg_{q,\eta}(u)\bigg].
\end{equation}
Next, consider the set of stopping times $\big\{\tau_{q,\eta}^{\kappa}\hspace{1pt},\,\kappa\geq0\big\}$ defined by
\begin{equation}\label{eq4.1.66}
\tau_{q,\eta}^{\kappa} = \inf\!\big\{t\geq0 \;|\; g_{q,\eta}(t)\geq\kappa\big\},\hspace{.75em} \tau_{q,\eta}^{0}=0,
\end{equation}
and note that they are finite, since
\begin{equation}\label{eq4.1.67}
\tau_{q,\eta}^{\kappa} \leq \frac{2\kappa(1-\lambda_{q})}{\eta(q-1)\sigma_{max}^{2}}\hspace{1pt},
\end{equation}
and strictly increasing, and that $g_{q,\eta}(\tau_{q,\eta}^{\kappa})=\kappa$ by continuity. Fix $\kappa>0$ and set $\tau=\tau_{q,\eta}^{\kappa}$ in \eqref{eq4.1.65}. Using an idea from \cite{Berkaoui:2008}, define a stochastic time-change $s=g_{q,\eta}(u)$ such that $u=\tau_{q,\eta}^{s}$ and note that
\begin{equation}\label{eq4.1.68}
g_{q,\eta}(T\wedge\tau_{q,\eta}^{\kappa}) = g_{q,\eta}(\tau_{q,\eta}^{\kappa})\wedge g_{q,\eta}(T) = \kappa\wedge g_{q,\eta}(T).
\end{equation}
By Lebesgue's change-of-time formula (see, e.g., Theorem A4.3 in \cite{Sharpe:1988}), we get
\begin{align}\label{eq4.1.69}
\E\bigg[\int_{0}^{T\wedge\hspace{.5pt}\tau_{q,\eta}^{\kappa}}\sup_{t\in[0,T]}|e^{x}_{t\wedge u}|^{q}\,dg_{q,\eta}(u)\bigg] &= \E\bigg[\int_{0}^{\kappa\wedge g_{q,\eta}(T)}\sup_{t\in[0,T]}\big|e^{x}_{t\wedge \tau_{q,\eta}^{s}}\big|^{q}\,ds\bigg] \nonumber\\[2pt]
&\leq \int_{0}^{\kappa}{\E\bigg[\sup_{t\in[0,T]}\big|e^{x}_{t\wedge\tau_{q,\eta}^{s}}\big|^{q}\bigg]ds}.
\end{align}
Substituting back into \eqref{eq4.1.65} with this upper bound yields
\begin{equation}\label{eq4.1.70}
\E\bigg[\sup_{t\in[0,T]}\big|e^{x}_{t\wedge\tau_{q,\eta}^{\kappa}}\big|^{q}\bigg] \leq C_{q,\eta}\bigg(\frac{N}{\log(2N)}\bigg)^{-\frac{q}{2}} + \int_{0}^{\kappa}{\E\bigg[\sup_{t\in[0,T]}\big|e^{x}_{t\wedge\tau_{q,\eta}^{s}}\big|^{q}\bigg]ds},
\end{equation}
and applying Gronwall's inequality leads to
\begin{equation}\label{eq4.1.71}
\E\bigg[\sup_{t\in[0,T]}\big|e^{x}_{t\wedge\tau_{q,\eta}^{\kappa}}\big|^{q}\bigg] \leq C_{q,\eta}e^{\kappa}\bigg(\frac{N}{\log(2N)}\bigg)^{-\frac{q}{2}},
\end{equation}
for all $\kappa>0$. Proceeding in a similar way as in the argument of \eqref{eq4.1.65} and setting $\tau=T$, we get
\begin{equation}\label{eq4.1.72}
\E\bigg[\sup_{t\in[0,T]}|e^{x}_{t}|^{p}\bigg] \leq C_{p,\eta}\bigg(\frac{N}{\log(2N)}\bigg)^{-\frac{p}{2}} + \E\bigg[\int_{0}^{T}\sup_{t\in[0,T]}|e^{x}_{t\wedge u}|^{p}\,dg_{p,\eta}(u)\bigg].
\end{equation}
However, note from \eqref{eq4.1.55} that both
\begin{equation}\label{eq4.1.73}
\frac{\sqrt{q}}{\lambda_{q}} = \sqrt{q} + \sqrt{q(1+2\beta_{0}^{-2})-2\beta_{0}^{-2}(1-2\Delta_{\sigma}^{-1})}
\end{equation}
and
\begin{equation}\label{eq4.1.74}
\frac{q}{1-\lambda_{q}} = q + \frac{q\sqrt{q}}{\sqrt{q(1+2\beta_{0}^{-2})-2\beta_{0}^{-2}(1-2\Delta_{\sigma}^{-1})}}
\end{equation}
are increasing in $q$, and hence that
\begin{equation}\label{eq4.1.75}
\frac{d}{dt}\hspace{1pt}g_{p,\eta}(t) \leq \frac{d}{dt}\hspace{1pt}g_{q,\eta}(t),\hspace{1em} \forall\hspace{.5pt}t\in[0,T].
\end{equation}
Using the same time-change from before, Fubini's theorem and H\"older's inequality, we deduce that
\begin{align}\label{eq4.1.76}
\E\bigg[\int_{0}^{T}\sup_{t\in[0,T]}|e^{x}_{t\wedge u}|^{p}\,dg_{p,\eta}(u)\bigg] &\leq \E\bigg[\int_{0}^{T}\sup_{t\in[0,T]}|e^{x}_{t\wedge u}|^{p}\,dg_{q,\eta}(u)\bigg] \nonumber\\[2pt]
&\leq \int_{0}^{\infty}{\E\bigg[\sup_{t\in[0,T]}\big|e^{x}_{t\wedge\tau_{q,\eta}^{s}}\big|^{p}\Ind_{s\leq g_{q,\eta}(T)}\bigg]ds} \nonumber\\[2pt]
&\leq \int_{0}^{\infty}{\E\bigg[\sup_{t\in[0,T]}\big|e^{x}_{t\wedge\tau_{q,\eta}^{s}}\big|^{q}\bigg]^{\frac{p}{q}}\Prob\Big(s\leq g_{q,\eta}(T)\Big)^{1-\frac{p}{q}}\,ds}.
\end{align}
Combining \eqref{eq4.1.71}, \eqref{eq4.1.72} and \eqref{eq4.1.76} yields
\begin{equation}\label{eq4.1.77}
\E\bigg[\sup_{t\in[0,T]}|e^{x}_{t}|^{p}\bigg] \leq \bigg(\frac{N}{\log(2N)}\bigg)^{-\frac{p}{2}}\bigg\{C_{p,\eta} + C_{q,\eta}^{\frac{p}{q}}\int_{0}^{\infty}{\exp\bigg\{\frac{sp}{q}\bigg\}\Prob\Big(s\leq g_{q,\eta}(T)\Big)^{1-\frac{p}{q}}\,ds}\bigg\}.
\end{equation}
All that is left to do is bound the probability on the right-hand side from above. For brevity, define
\begin{equation}\label{eq4.1.78}
w_{p,q} = \frac{p}{q-p}
\end{equation}
and let $w>w_{p,q}$. Markov's inequality yields
\begin{align}\label{eq4.1.79}
\Prob\Big(s\leq g_{q,\eta}(T)\Big) \leq \exp\left\{-w s\right\}\E\Big[\exp\big\{w g_{q,\eta}(T)\big\}\Big].
\end{align}
From \eqref{eq4.1.64} and H\"older's inequality, we get
\begin{align}\label{eq4.1.80}
\E\Big[\exp\big\{w g_{q,\eta}(T)\big\}\Big] &\leq \exp\left\{w a_{q,\eta}T\right\}\E\bigg[\exp\bigg\{w(1+\eta)\big(f_{q}(\lambda_{q})+\eta b_{q,\eta}\big)\int_{0}^{T}{v_{t}\,dt}\bigg\}\bigg]^{\frac{1}{1+\eta}} \nonumber\\[2pt]
&\times \E\bigg[\exp\bigg\{w\eta(1+\eta)c_{q,\eta}\int_{0}^{T}{v_{t}^{-1}\,dt}\bigg\}\bigg]^{\frac{\eta}{1+\eta}}.
\end{align}
However, note that
\begin{equation}\label{eq4.1.81}
w(1+\eta)\big(f_{q}(\lambda_{q})+\eta b_{q,\eta}\big) = w_{p,q}f_{q}(\lambda_{q}) + (w-w_{p,q})f_{q}(\lambda_{q}) + \eta w\big(f_{q}(\lambda_{q})+(1+\eta)b_{q,\eta}\big).
\end{equation}
Using exponential integrability properties of the CIR process from Lemma \ref{Lem3.2} together with \eqref{eq4.1.37.4} and a continuity argument, since $\nu>1$, we conclude that there exist $\eta$ sufficiently small and $w$ sufficiently close to $w_{p,q}$ such that the two expectations on the right-hand side of \eqref{eq4.1.80}, and hence the one on the left-hand side, are finite. Finally, substituting back into \eqref{eq4.1.77} with \eqref{eq4.1.79}, since
\begin{align}\label{eq4.1.82}
\int_{0}^{\infty}{\exp\bigg\{\frac{sp}{q}-w s\bigg(1-\frac{p}{q}\bigg)\bigg\}\,ds} &= \int_{0}^{\infty}{\exp\bigg\{-\frac{(w-w_{p,q})(q-p)s}{q}\bigg\}\,ds} \nonumber\\[2pt]
&= \frac{q}{(w-w_{p,q})(q-p)}\hspace{1pt},
\end{align}
we deduce that
\begin{equation}\label{eq4.1.83}
\E\bigg[\sup_{t\in[0,T]}|e^{x}_{t}|^{p}\bigg] \leq \bigg(\frac{N}{\log(2N)}\bigg)^{-\frac{p}{2}}\bigg\{C_{p,\eta} + \frac{q}{(w-w_{p,q})(q-p)}\hspace{1pt}C_{q,\eta}^{\frac{p}{q}}\E\Big[\exp\big\{w g_{q,\eta}(T)\big\}\Big]^{1-\frac{p}{q}}\bigg\},
\end{equation}
and the conclusion follows.
\end{proof}

\subsection{Moment bounds}\label{subsec:moment}

Many models with stochastic volatility dynamics have the undesirable feature that moments of order higher than 1 can explode in finite time \cite{Andersen:2007}. The finiteness of moments of order higher than 1 of the exact and numerical solutions of a SDE is an important ingredient in the convergence analysis \cite{Higham:2002}. The finiteness of higher moments was established in \cite{Cozma:2016a} for explicit Euler approximations to stochastic-local volatility models. We extend this result to stochastic path-dependent volatility models and include the proof here for completeness.

\begin{lemma}\label{Lem4.6}
Suppose that Assumption \ref{Asm2.1} holds and let $p>1$.
\begin{enumerate}[(1)]
\item{%
If $T<T^{\scalebox{0.6}{\emph{\text{CIR}}}}_{S}(p)$, then the spot process has a bounded $p$th moment, i.e.,
\begin{equation}\label{eq4.1.84}
\sup_{t\in[0,T]}\E\big[S_{t}^{p}\big] < \infty,
\end{equation}
where
\begin{equation}\label{eq4.1.85}
T^{\scalebox{0.6}{\emph{\text{CIR}}}}_{S}(p) = \frac{2}{\sqrt{\smash[b]{(\phi(p)-k^{2})^{+}}}}\Bigg[\frac{\pi}{2} + \arctan\Bigg(\frac{k}{\sqrt{\smash[b]{(\phi(p)-k^{2})^{+}}}}\Bigg)\Bigg],
\end{equation}
with $\phi$ given in \eqref{eq2.3.8}.
}
\item{%
If $\nu\geq1$ and $T<T^{*}_{S}(p)$, with $T^{\scalebox{0.6}{\emph{\text{FTE}}}}_{S}$ and $T^{\scalebox{0.6}{\emph{\text{BEM}}}}_{S}$ given in \eqref{eq2.3.6} and \eqref{eq2.3.7}, respectively, then the approximated spot process has a bounded $p$th moment, i.e., there exists $N_{0}\in\mathbb{N}$ such that
\begin{equation}\label{eq4.1.87}
\sup_{N>N_{0}}\,\sup_{t\in[0,T]}\E\big[\bar{S}_{t}^{p}\big] < \infty.
\end{equation}
}
\end{enumerate}
\end{lemma}
\begin{proof}
See Appendix \ref{sec:aux3}.
\end{proof}

\subsection{Proof of Theorem \ref{Thm2.7}}\label{subsec:proof}

With these results at our disposal, we are now ready to prove the main theorem.

\begin{proof}
By a continuity argument, we can find $2\vee p<q<p^{*}(\nu)$ such that
\begin{equation}\label{eq4.1.94}
T < T_{x}^{*}(q) \wedge\hspace{.5pt} T_{S}^{*}\big(pq(q-p)^{-1}\big).
\end{equation}
Fix $r>1$ and recall the definition of $\phi$ from \eqref{eq2.3.8}. On the one hand, if $\phi(r)\geq4k^{2}$, then
\begin{equation}\label{eq4.1.95}
\frac{2}{\sqrt{\phi(r)-k^{2}}}\Bigg[\frac{\pi}{2} + \arctan\Bigg(\frac{k}{\sqrt{\phi(r)-k^{2}}}\Bigg)\Bigg] \geq	
\frac{1}{\sqrt{\phi(r)-k^{2}}}\hspace{1pt}\sqrt{\frac{\sqrt{\phi(r)}+k}{\sqrt{\phi(r)}-k}} = \frac{1}{\sqrt{\phi(r)}-k} \geq \frac{1}{\sqrt{\phi(r)}}\hspace{1pt}.
\end{equation}
On the other hand, if $k^{2}<\phi(r)<4k^{2}$, then
\begin{equation}\label{eq4.1.96}
\frac{2}{\sqrt{\phi(r)-k^{2}}}\Bigg[\frac{\pi}{2} + \arctan\Bigg(\frac{k}{\sqrt{\phi(r)-k^{2}}}\Bigg)\Bigg] \geq
\frac{4}{\sqrt{\phi(r)-k^{2}}}\hspace{1pt}\sqrt{\frac{k^{2}}{\phi(r)}\bigg(1-\hspace{1pt}\frac{k^{2}}{\phi(r)}\bigg)} = \frac{4k}{\phi(r)} \geq \frac{1}{\sqrt{\phi(r)}}\hspace{1pt}.
\end{equation}
Therefore, we have that $T^{\scalebox{0.6}{\text{CIR}}}_{S}(r) \geq T^{\scalebox{0.6}{\text{FTE}}}_{S}(r) \geq T^{\scalebox{0.6}{\text{BEM}}}_{S} (r)$ for all $r>1$. From the Mean-Value Theorem and H\"older's inequality, we deduce that
\begin{align}\label{eq4.1.97}
\sup_{t\in[0,T]}\E\Big[\big|S_{t}-\bar{S}_{t}\big|^{p}\Big]^{\frac{1}{p}} &\leq \sup_{t\in[0,T]}\E\Big[\max\!\left\{S_{t}^{p},\bar{S}_{t}^{p}\right\}|x_{t}-\bar{x}_{t}|^{p}\Big]^{\frac{1}{p}} \nonumber\\[2pt]
&\leq \bigg\{\sup_{t\in[0,T]}\E\Big[S_{t}^{\frac{pq}{q-p}}\Big]^{\frac{q-p}{pq}} + \sup_{N>N_{0}}\,\sup_{t\in[0,T]}\E\Big[\bar{S}_{t}^{\frac{pq}{q-p}}\Big]^{\frac{q-p}{pq}}\bigg\} \nonumber\\[0pt]
&\times \sup_{t\in[0,T]}\E\Big[|x_{t}-\bar{x}_{t}|^{q}\Big]^{\frac{1}{q}},
\end{align}
for some $N_{0}\in\mathbb{N}$ suitably chosen. The conclusion follows from Proposition \ref{Prop4.5} and Lemma \ref{Lem4.6}.
\end{proof}

\section{Numerical tests}\label{sec:numerics}

In this section, we assume the spot process dynamics from \eqref{eq2.1.1} with $\mu=0$ and perform a numerical analysis of the strong and weak convergence of the approximation process with the LE--FTE scheme, i.e., when the LE and the FTE schemes are employed in the discretization of the spot process and its squared volatility, respectively. Throughout this section, we fix the time horizon $T=1$ and assign the following values to the underlying model parameters as a base case, and vary a selection individually:
\begin{equation}\label{eq5.1.1}
S_{0}=1,\quad v_{0}=0.025,\quad k=8,\quad \theta=0.02,\quad \xi=0.2,\quad \rho=-0.1.
\end{equation}
These values are consistent with empirical observations in equity and FX markets and are close to the calibrated values in Table 2 in \cite{Cozma:2017a} and Table 1 in \cite{Cozma:2017b}.

\subsection{Strong convergence}\label{subsec:strong}

We define on $\mathcal{D}=\big\{(t,x,y)\in[0,T]\!\times\!\mathbb{R}^{2}_{+} \,|\, S_{min}\leq x\leq S_{0}\vee x\leq y\leq S_{max}\big\}$ a parametric leverage function $\sigma$ and extrapolate it flat outside these bounds, where
\begin{equation}\label{eq5.1.2}
S_{min} = S_{0}e^{-\frac{6}{2}\sqrt{v_{0}\hspace{.5pt}T}} \hspace{1em}\text{ and }\hspace{1em} S_{max} = S_{0}e^{\frac{6}{2}\sqrt{v_{0}\hspace{.5pt}T}}.
\end{equation}
In particular, we use a stochastic volatility inspired (SVI) parameterization (see \cite{Gatheral:2006,Hambly:2016}) in both spot and running maximum, i.e.,
\begin{align}\label{eq5.1.3}
\sigma(t,x,y) &= \frac{1}{2}\Big[\sigma_{1}\big(t+1,\log(S_{min}\vee x\wedge S_{max})-\log(S_{0})\big) \nonumber\\[3pt]
&+\sigma_{2}\big(t+1,\log(S_{min}\vee y\wedge S_{max})-\log(S_{0})\big)\Big],
\end{align}
where, for all $i\in\{1,2\}$,
\begin{equation}\label{eq5.1.4}
\sigma_{i}(u,z) = \frac{1}{\sqrt{u}}\hspace{1pt}\sqrt{a_{i}+b_{i}\Big(c_{i}(z-d_{i})+\sqrt{(z-d_{i})^{2}+e_{i}^{2}}\hspace{1pt}\Big)}\hspace{1pt}.
\end{equation}
We assign the following values to the parameters:
\begin{equation}\label{eq5.1.5}
a_{1,2}=1,\quad b_{1,2}=2,\quad c_{1,2}=0,\quad d_{1,2}=0,\quad e_{1,2}=0.25.
\end{equation}
In Figure \ref{fig:1}, we plot the leverage function $\sigma$ at three different time slices: $t=0$, $t=0.5$ and $t=1$. Note that this leverage function is constant outside a bounded interval by definition. Furthermore, one can easily show that $\sigma$ is 1/2-H\"older continuous in time and Lipschitz continuous in spot and running maximum, and hence that Assumptions \ref{Asm2.4} and \ref{Asm2.5} are satisfied. From Proposition \ref{Prop2.6}, we conclude that Assumptions \ref{Asm2.1} and \ref{Asm2.2} are also satisfied.
\begin{figure}[htb]
\begin{center}
\captionsetup{justification=centering}
\includegraphics[width=1.0\linewidth,height=3.2in]{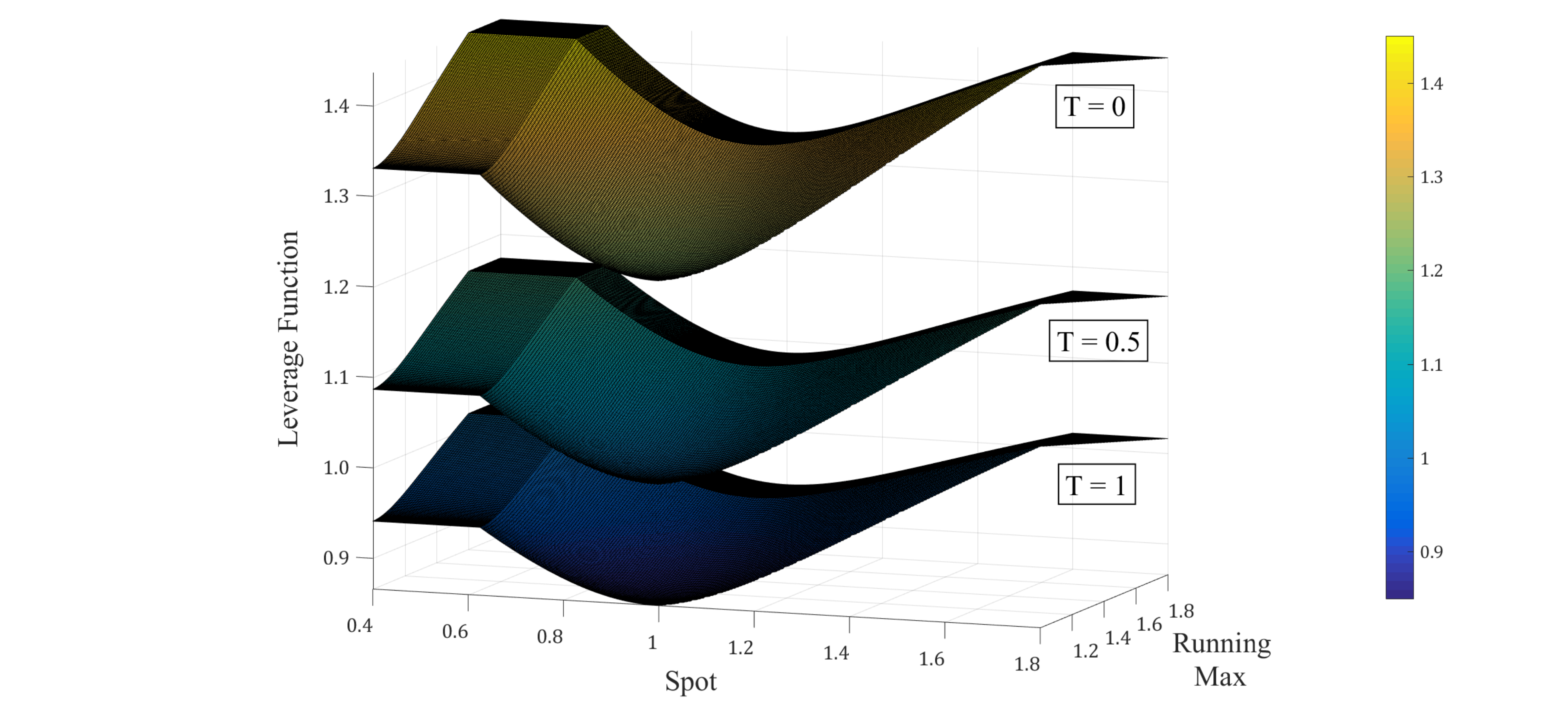}
\end{center}
\caption{The leverage function $\sigma$ with the SVI parameterization plotted against the spot and the running maximum at three different time slices.}
\label{fig:1}
\end{figure}

In order to establish the strong convergence in $L^{p}$ with order 1/2 (up to a logarithmic factor) of the approximation process, we first need to compute the critical time $T^{\scalebox{0.6}{\text{FTE}}}(p)$ from \eqref{eq2.3.9}. Recall the definitions of $\sigma_{max}$, $C_{\sigma,x}$ and $C_{\sigma,m}$ from \eqref{eq4.1.23} -- \eqref{eq4.1.29}. A straightforward technical analysis of the leverage function yields $\sigma_{max}=1.437$, $C_{\sigma,x}=0.307$ and $C_{\sigma,m}=0.307$. Therefore, we obtain $T^{\scalebox{0.6}{\text{FTE}}}(1)=132.58$ and $T^{\scalebox{0.6}{\text{FTE}}}(2)=12.57$, both greater than $T=1$, and hence all conditions in the statement of Theorem \ref{Thm2.7} are satisfied. For illustration, we plot in Figure \ref{fig:2} the critical time against the power $p$ (in the $L^{p}$ norm) and the mean reversion rate $k$ of the squared volatility process. First, we infer from Figures \ref{fig:2a} and \ref{fig:2b} that $\lim_{p\to p^{\scalebox{0.6}{\text{FTE}}}}T^{\scalebox{0.6}{\text{FTE}}}(p)=0$, a fact which can easily be verified from the definition of the critical time. Second, we infer from Figures \ref{fig:2c} and \ref{fig:2d} that $\lim_{k\to \infty}T^{\scalebox{0.6}{\text{FTE}}}(p)=\infty$. The limiting case corresponds to a purely path-dependent volatility, where the strong convergence result holds for all $T>0$.
\begin{figure}[htb]
\centering
\begin{subfigure}{.5\textwidth}
  \centering
	\captionsetup{justification=centering}
  \includegraphics[width=.98\linewidth,height=1.72in]{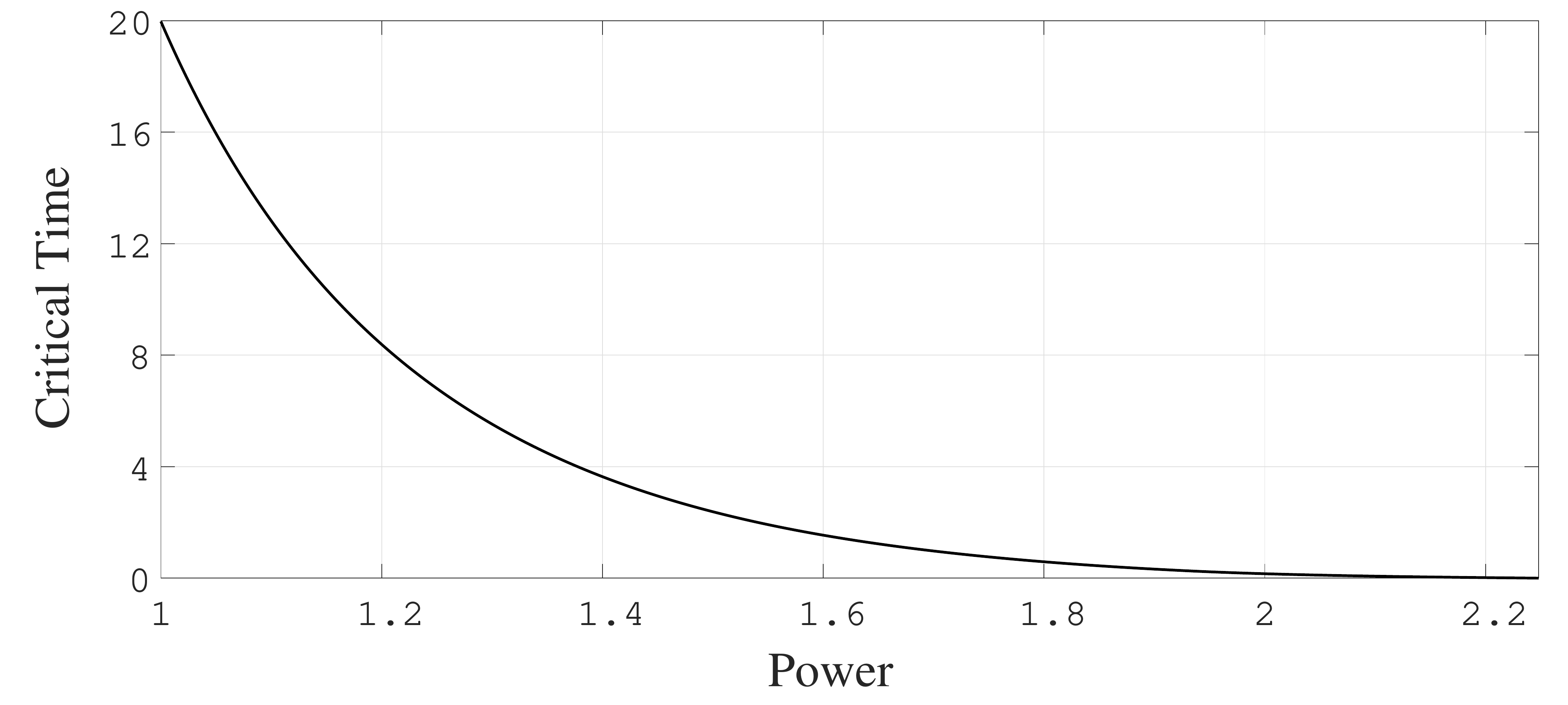}
  \caption{$k=4$}
  \label{fig:2a}
\end{subfigure}%
\begin{subfigure}{.5\textwidth}
  \centering
	\captionsetup{justification=centering}
  \includegraphics[width=.98\linewidth,height=1.72in]{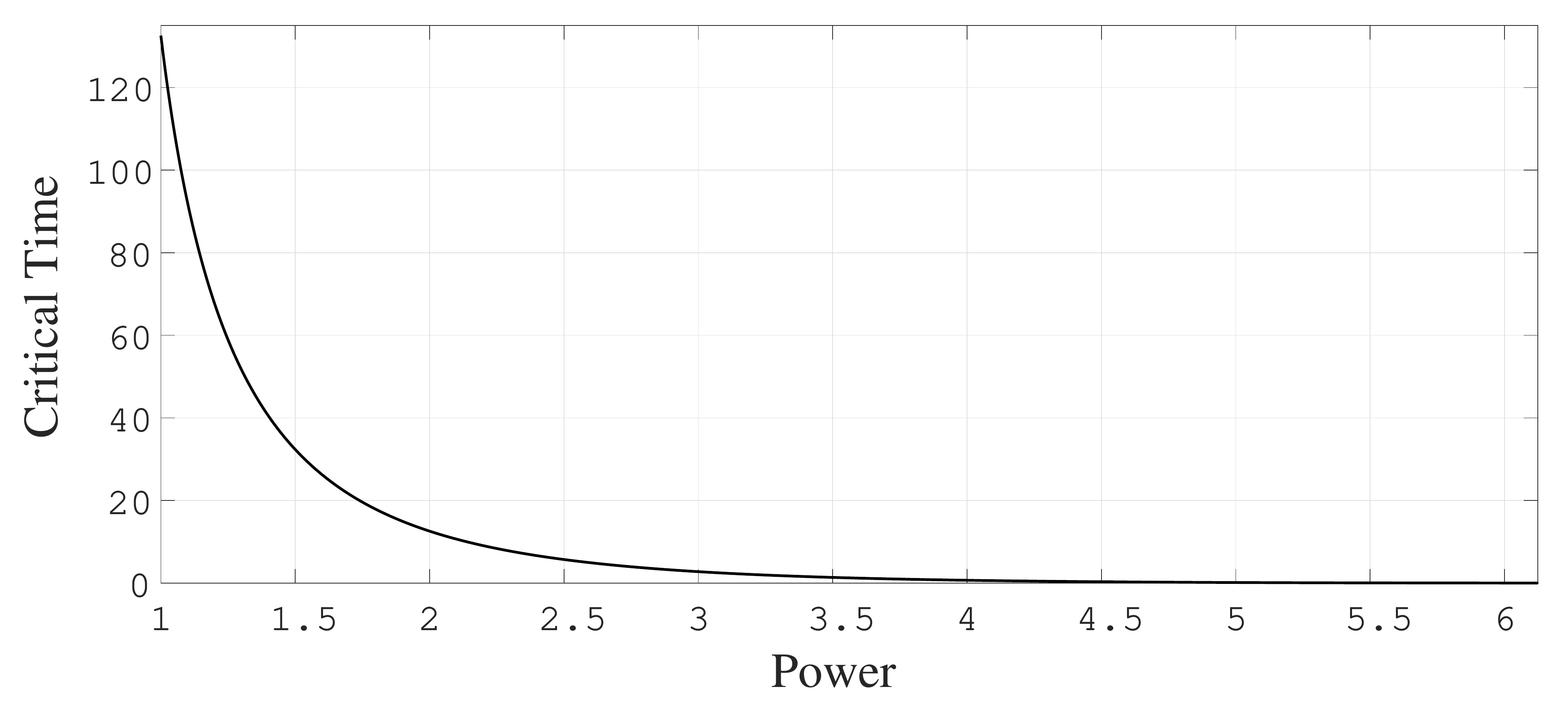}
  \caption{$k=8$}
  \label{fig:2b}
\end{subfigure} \\[.5em]
\begin{subfigure}{.5\textwidth}
  \centering
	\captionsetup{justification=centering}
  \includegraphics[width=.98\linewidth,height=1.72in]{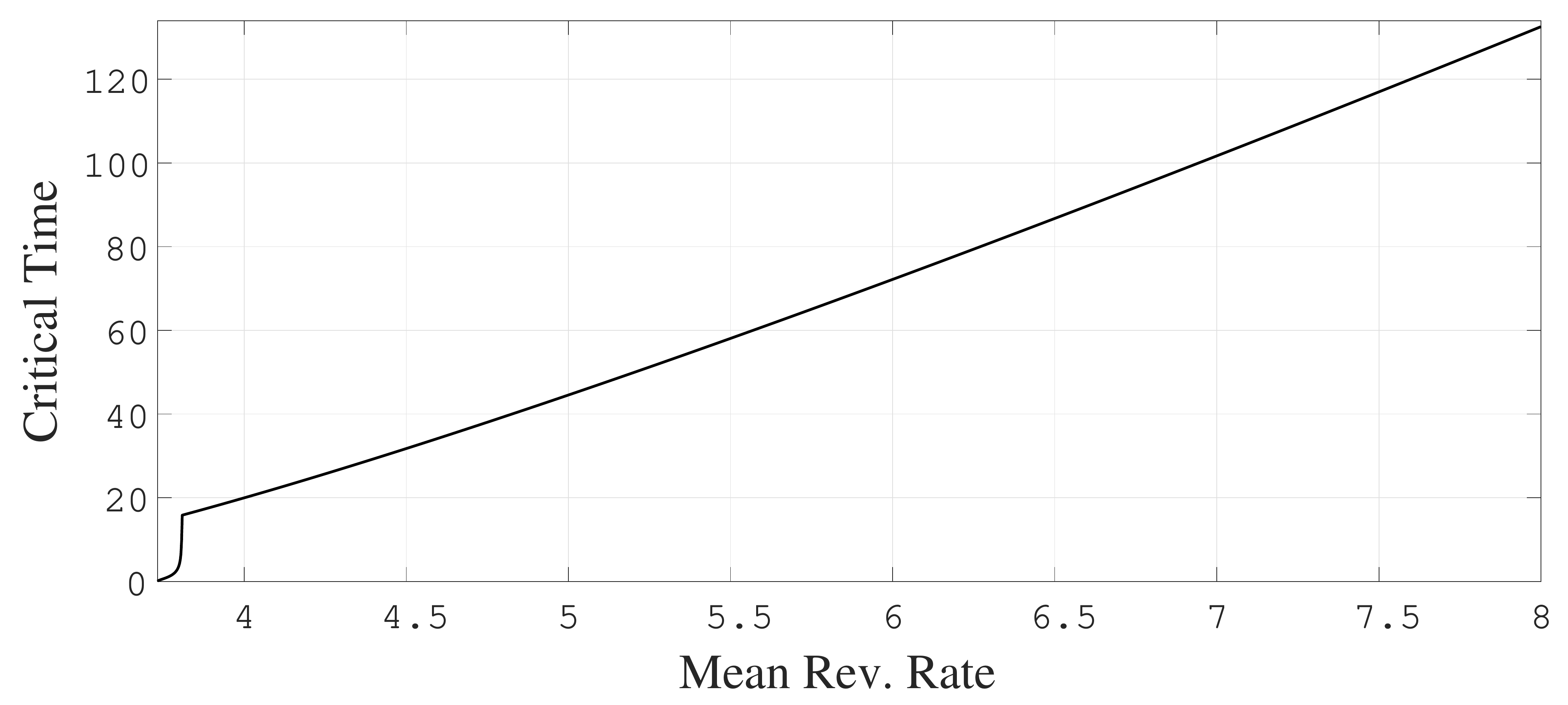}
  \caption{$p=1$}
  \label{fig:2c}
\end{subfigure}%
\begin{subfigure}{.5\textwidth}
  \centering
	\captionsetup{justification=centering}
  \includegraphics[width=.98\linewidth,height=1.72in]{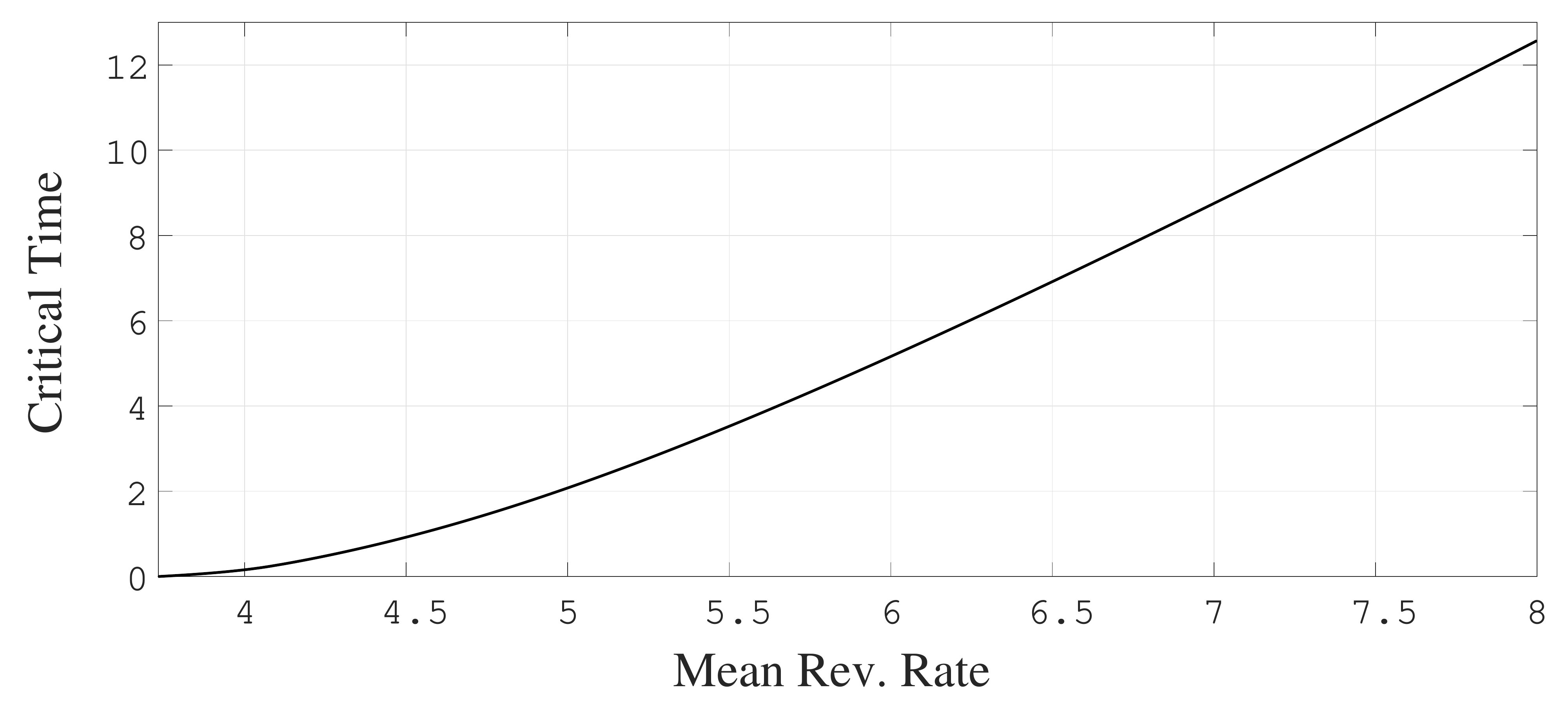}
  \caption{$p=2$}
  \label{fig:2d}
\end{subfigure}
\caption{The critical time defined in \eqref{eq2.3.9} plotted against the power and the mean reversion rate when $k\in\{4,8\}$ and $p\in\{1,2\}$, respectively.}
\label{fig:2}
\end{figure}

Next, we denote by $\bar{S}_{T\hspace{-1pt},\hspace{1pt}N}$ the value at time $T$ of the approximation process corresponding to an equidistant discretization with $N$ time steps, and study the $L^{p}$ error
\begin{equation}\label{eq5.1.6}
\varepsilon_{\scalebox{0.6}{S}}(N) = \E\Big[\big|S_{T}-\bar{S}_{T\hspace{-1pt},\hspace{1pt}N}\big|^{p}\Big]^{\frac{1}{p}}
\end{equation}
when $p=1$ (convergence in $L^{1}$ implies weak convergence for a large class of options, see, e.g., \cite{Cozma:2016a}) and $p=2$ (convergence in $L^{2}$ is useful for multilevel Monte Carlo methods, see, e.g., \cite{Giles:2008}). Due to the difficulty in computing the quantity in \eqref{eq5.1.6}, we use Proposition \ref{Prop5.1} and estimate as proxy the difference between the values of the approximation process corresponding to $N$ time steps ($\bar{S}_{T\hspace{-1pt},\hspace{1pt}N}$) and $2N$ time steps ($\bar{S}_{T\hspace{-1pt},\hspace{1pt}2N}$) for the same Brownian path. A proof of Proposition \ref{Prop5.1} for $p=1$ can be found, for instance, in \cite{Alfonsi:2005}. However, since the extension to $p\geq1$ is non-trivial, we include the proof for the general case here.

\begin{proposition}\label{Prop5.1}
Let $T>0$ and $p\geq1$, and suppose that there exists $\eta>1-\frac{1}{p}$ such that
\begin{equation}\label{eq5.1.7}
\E\Big[\big|S_{T}-\bar{S}_{T\hspace{-1pt},\hspace{1pt}N}\big|^{p}\Big]^{\frac{1}{p}} = \mathcal{O}\Big(\big(\log(2N)\big)^{-\eta}\Big).
\end{equation}
Then, for any $\alpha>0$ and $\beta\geq0$,
\begin{equation}\label{eq5.1.8}
\E\Big[\big|S_{T}-\bar{S}_{T\hspace{-1pt},\hspace{1pt}N}\big|^{p}\Big]^{\frac{1}{p}} = \mathcal{O}\left(\frac{\big(\log(2N)\big)^{\beta}}{N^{\alpha}}\right) \hspace{1em}\Leftrightarrow\hspace{1em}
\E\Big[\big|\bar{S}_{T\hspace{-1pt},\hspace{1pt}N}-\bar{S}_{T\hspace{-1pt},\hspace{1pt}2N}\big|^{p}\Big]^{\frac{1}{p}} = \mathcal{O}\left(\frac{\big(\log(2N)\big)^{\beta}}{N^{\alpha}}\right).
\end{equation}
\end{proposition}
\begin{proof}
See Appendix \ref{sec:aux4}.
\end{proof}

In Theorem 1 in \cite{Hefter:2017}, a lower error bound was established for all discretization schemes for the CIR process based on equidistant evaluations of the Brownian motion in the accessible boundary regime. As a consequence of this result, the FTE scheme achieves at most a strong convergence order of $\nu$ when $\nu<1/2$. In fact, we demonstrated the $L^{1}$ order of $\min\{\nu,1/2\}$ of the FTE scheme for the CIR process numerically in \cite{Cozma:2017c}. Therefore, due to the CIR dynamics driving the squared volatility of the spot process in \eqref{eq2.1.1}, we expect a strong convergence order of the LE--FTE scheme strictly less than $1/2$ when $\nu<1/2$. The data in Figures \ref{fig:3} and \ref{fig:4} suggest an empirical $L^{1}$ (and $L^{2}$) order between $0$ and $1/2$ when $\nu<1/2$ and an order of $1/2$ when $\nu\geq1/2$, which is in line with the previous observation and also with our theoretical results when $\nu>2+\sqrt{3}$.
\begin{figure}[!htb]
\centering
\begin{subfigure}{.5\textwidth}
  \centering
	\captionsetup{justification=centering}
  \includegraphics[width=.98\linewidth,height=1.95in]{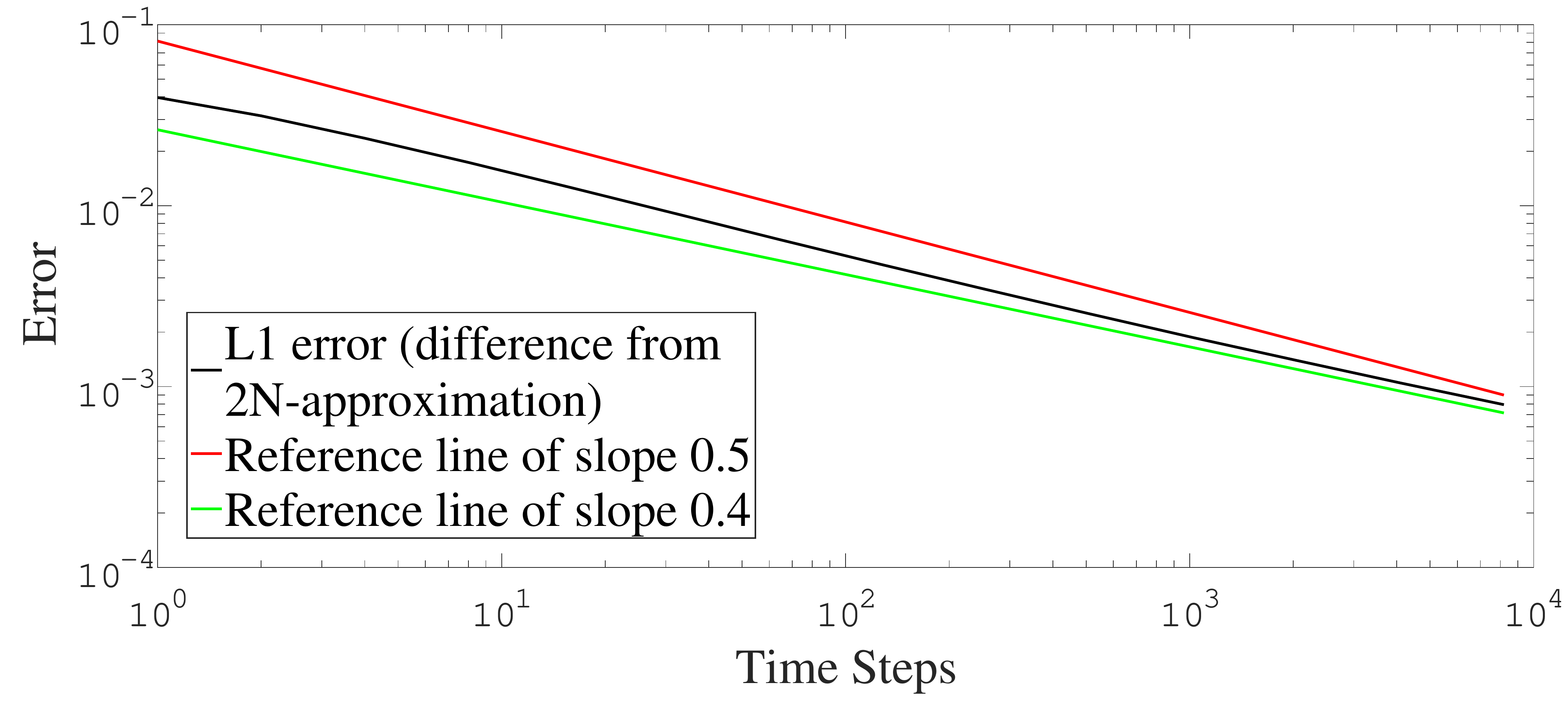}
  \caption{$\nu=0.25$}
  \label{fig:3a}
\end{subfigure}%
\begin{subfigure}{.5\textwidth}
  \centering
	\captionsetup{justification=centering}
  \includegraphics[width=.98\linewidth,height=1.95in]{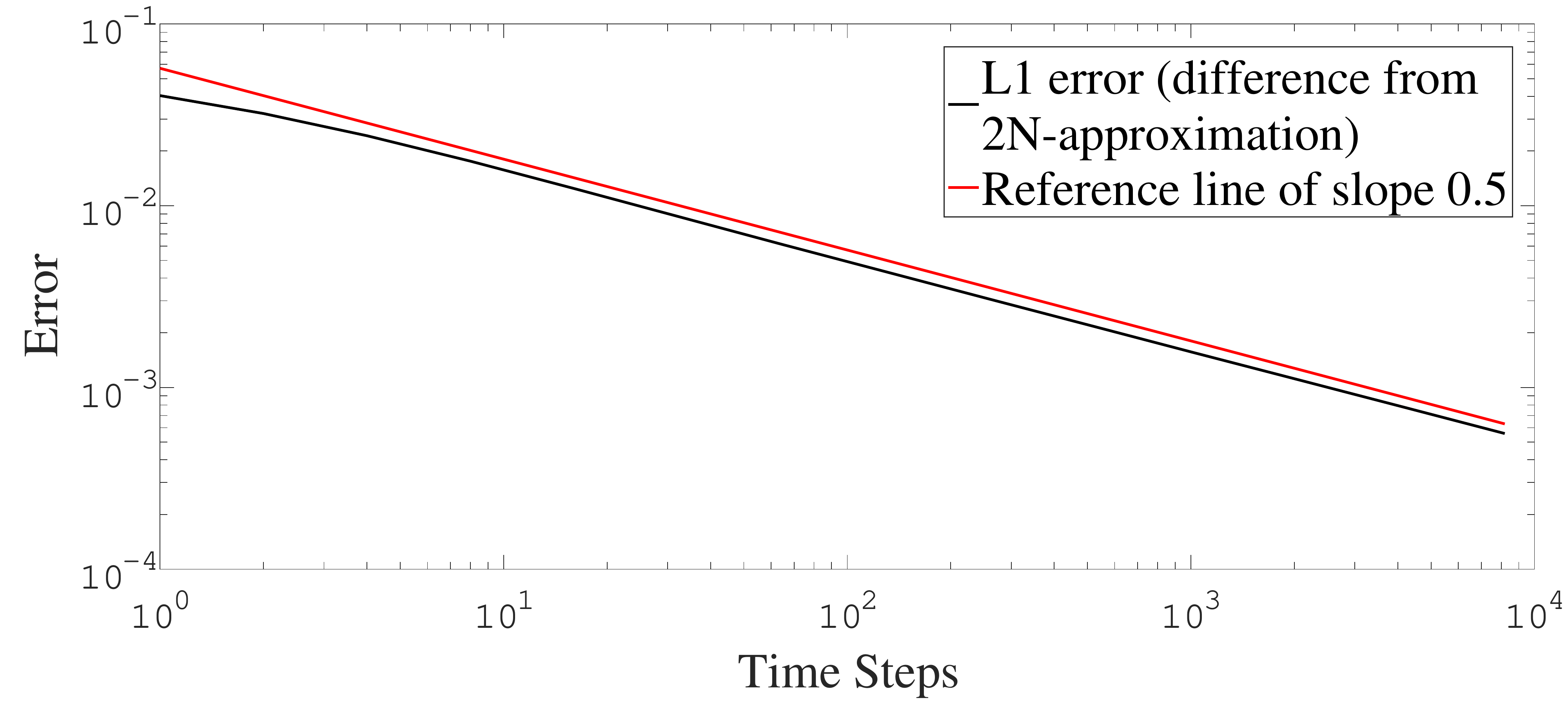}
  \caption{$\nu=0.5$}
  \label{fig:3b}
\end{subfigure} \\[.5em]
\begin{subfigure}{.5\textwidth}
  \centering
	\captionsetup{justification=centering}
  \includegraphics[width=.98\linewidth,height=1.95in]{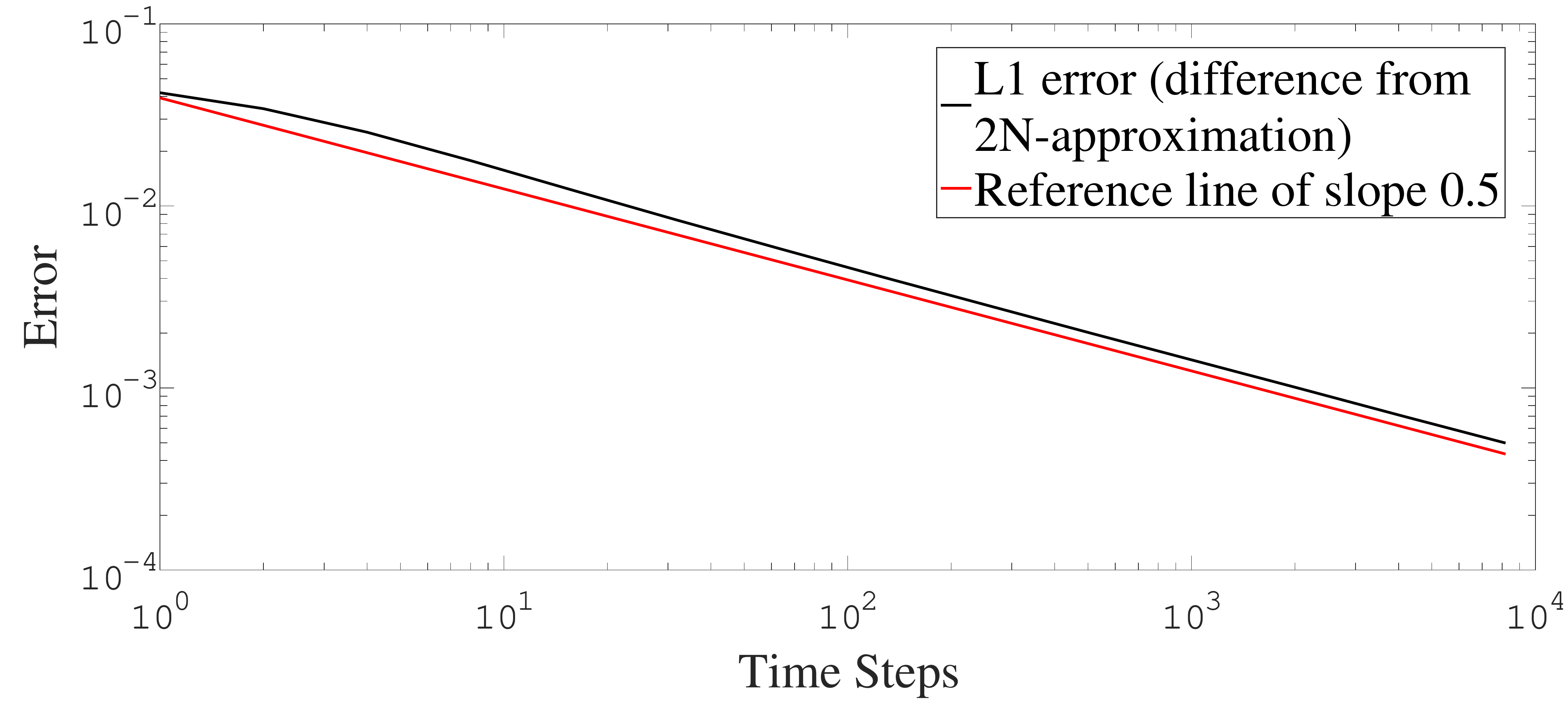}
  \caption{$\nu=1$}
  \label{fig:3c}
\end{subfigure}%
\begin{subfigure}{.5\textwidth}
  \centering
	\captionsetup{justification=centering}
  \includegraphics[width=.98\linewidth,height=1.95in]{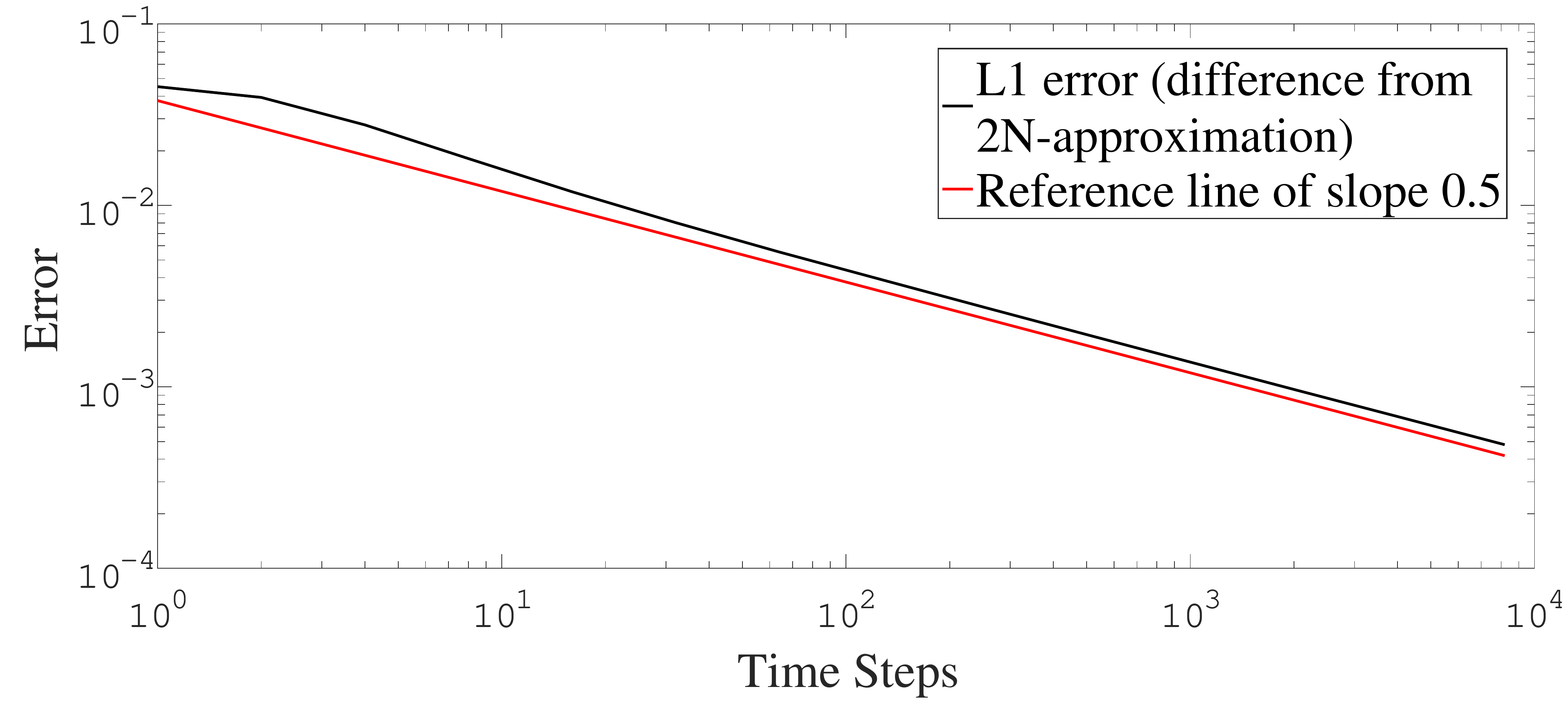}
  \caption{$\nu=2$}
  \label{fig:3d}
\end{subfigure} \\[.5em]
\begin{subfigure}{.5\textwidth}
  \centering
	\captionsetup{justification=centering}
  \includegraphics[width=.98\linewidth,height=1.95in]{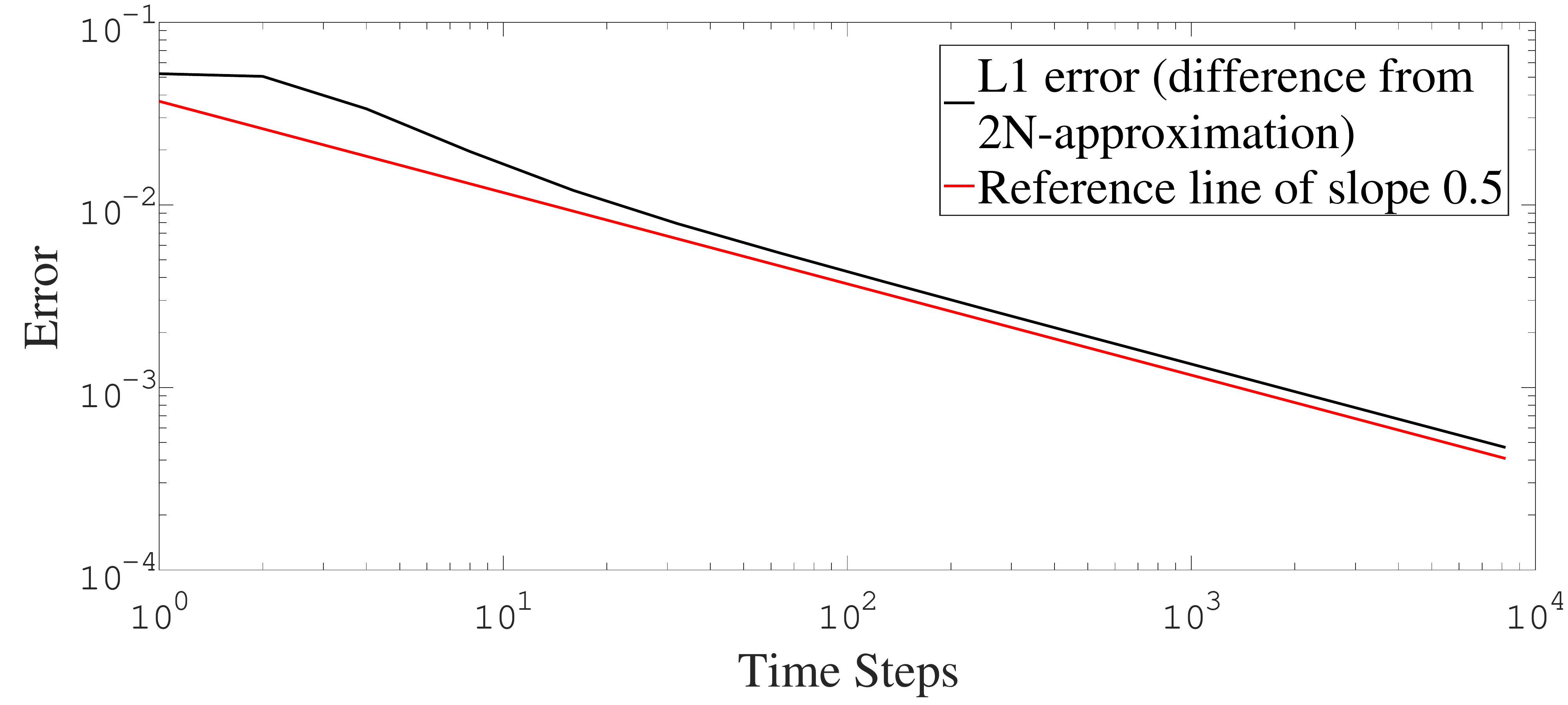}
  \caption{$\nu=4$}
  \label{fig:3e}
\end{subfigure}%
\begin{subfigure}{.5\textwidth}
  \centering
	\captionsetup{justification=centering}
  \includegraphics[width=.98\linewidth,height=1.95in]{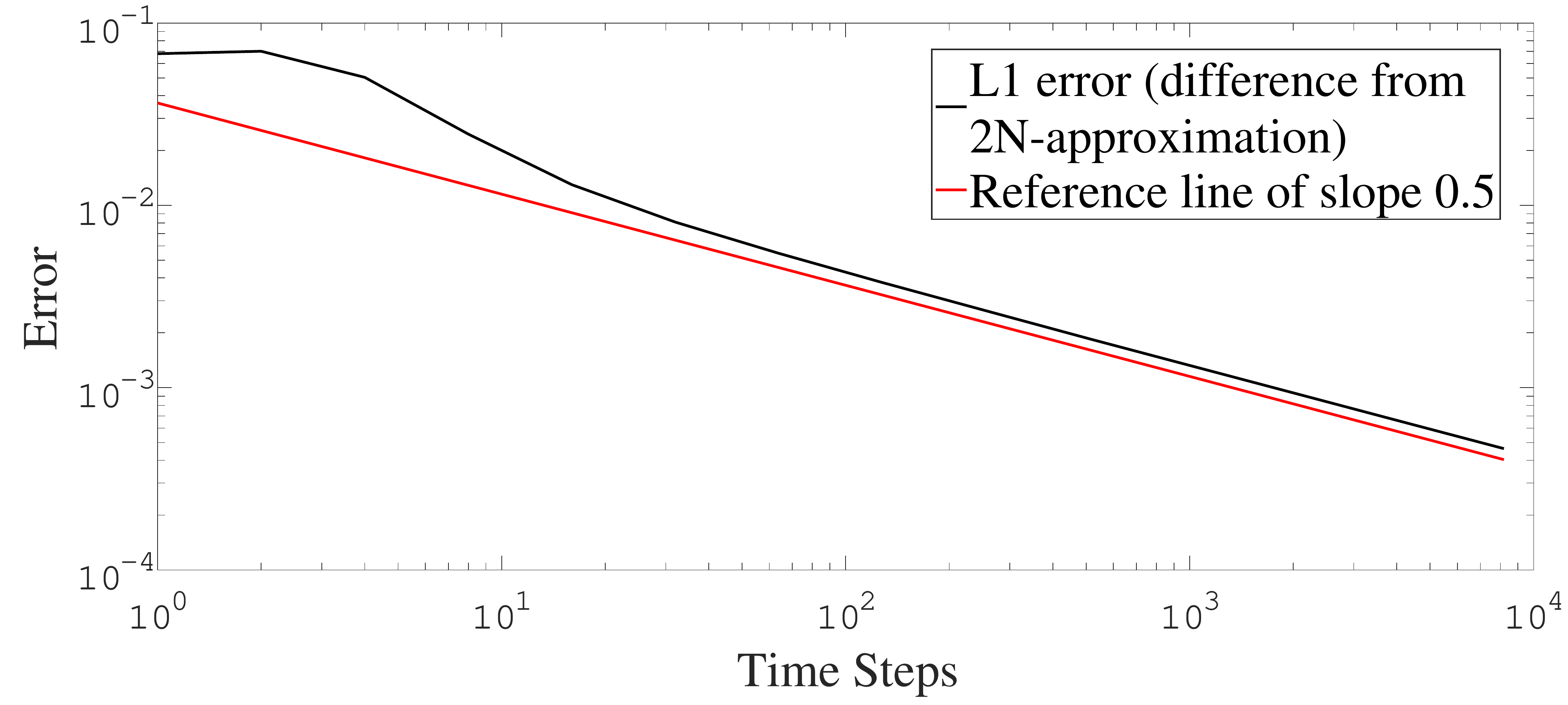}
  \caption{$\nu=8$}
  \label{fig:3f}
\end{subfigure}%
\caption{The $L^{1}$ errors against the number of time steps when $k\in\{0.25,0.5,1,2,4,8\}$ and the other parameters are as defined in \eqref{eq5.1.1}, computed using up to $2.6\!\times\!10^{6}$ Monte Carlo paths (for a relative error less than 10bp).}
\label{fig:3}
\end{figure}

\begin{figure}[!htb]
\centering
\begin{subfigure}{.5\textwidth}
  \centering
	\captionsetup{justification=centering}
  \includegraphics[width=.98\linewidth,height=1.95in]{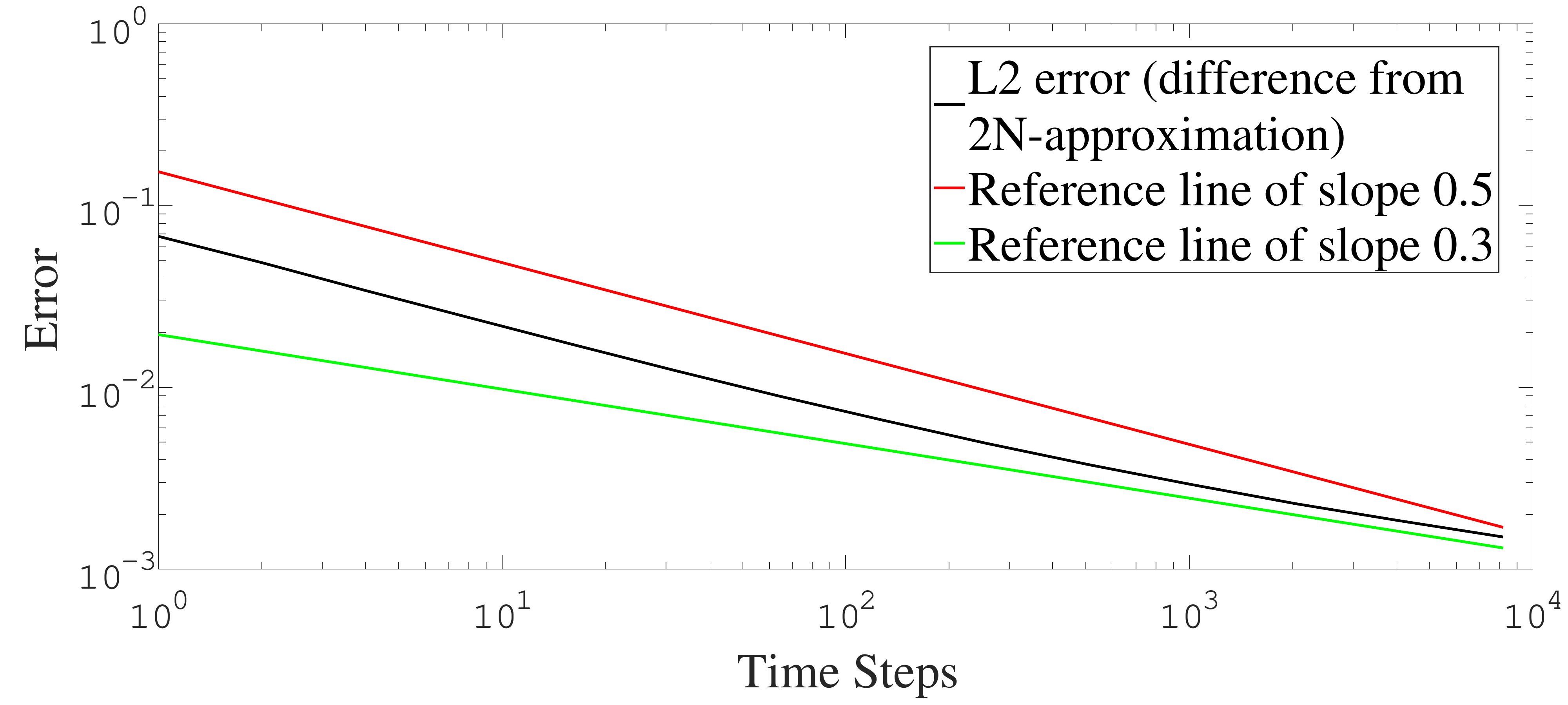}
  \caption{$\nu=0.25$}
  \label{fig:4a}
\end{subfigure}%
\begin{subfigure}{.5\textwidth}
  \centering
	\captionsetup{justification=centering}
  \includegraphics[width=.98\linewidth,height=1.95in]{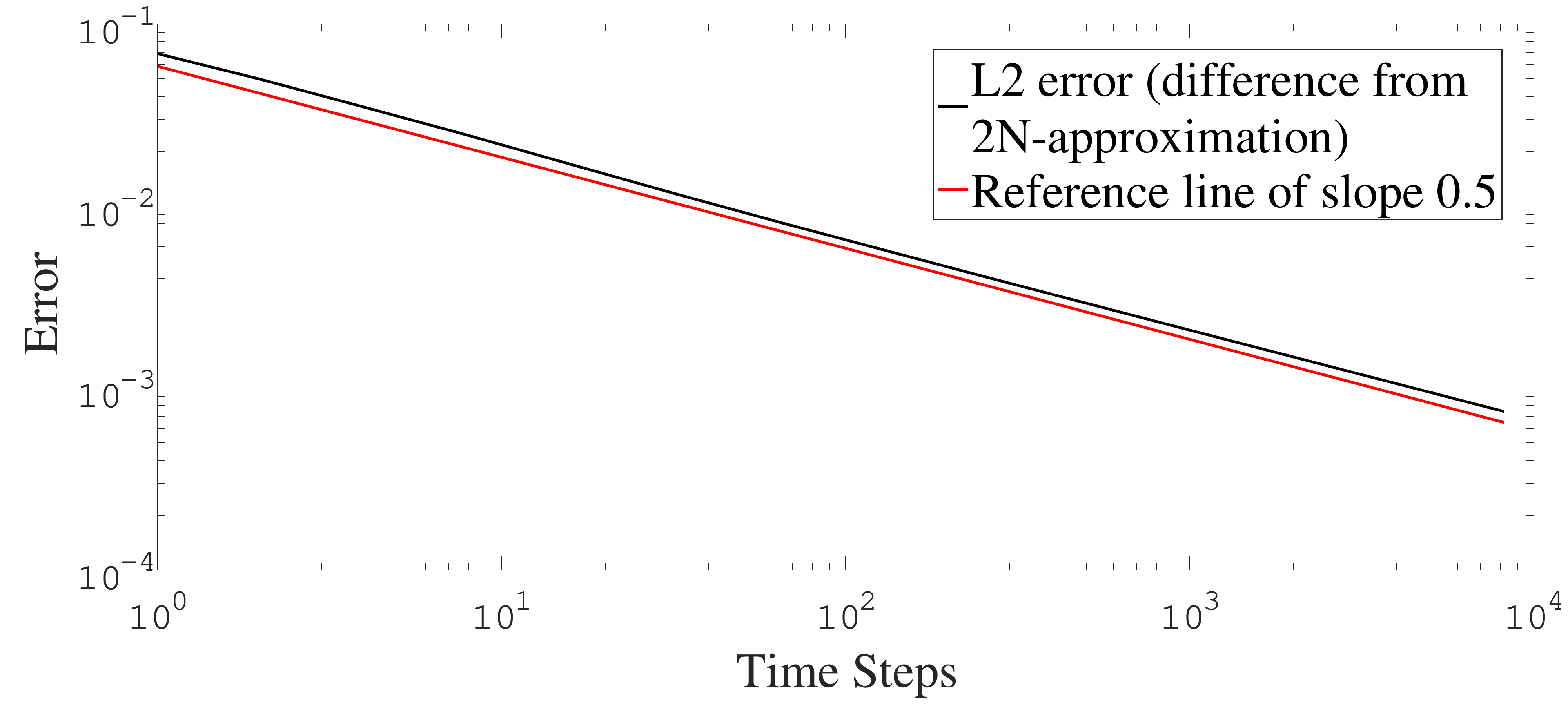}
  \caption{$\nu=0.5$}
  \label{fig:4b}
\end{subfigure} \\[.5em]
\begin{subfigure}{.5\textwidth}
  \centering
	\captionsetup{justification=centering}
  \includegraphics[width=.98\linewidth,height=1.95in]{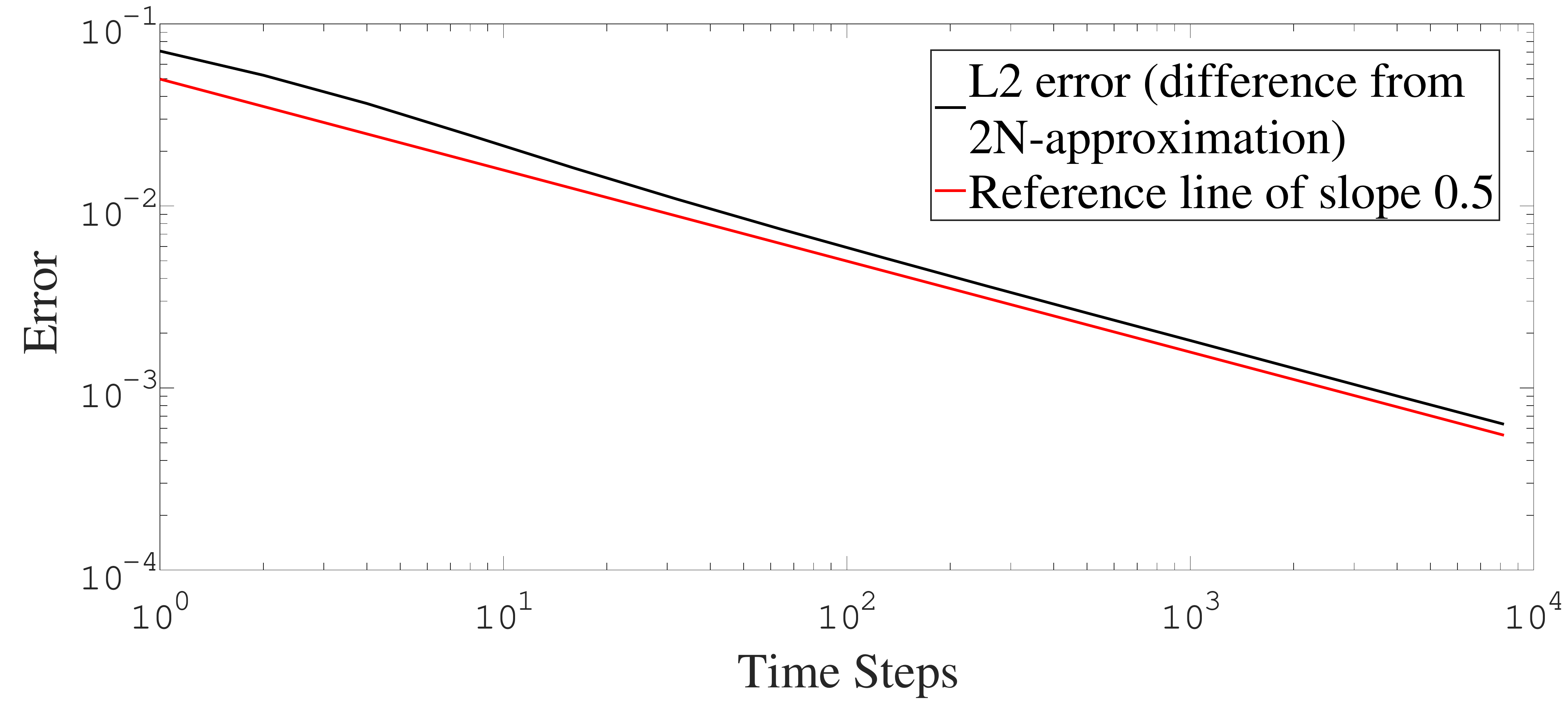}
  \caption{$\nu=1$}
  \label{fig:4c}
\end{subfigure}%
\begin{subfigure}{.5\textwidth}
  \centering
	\captionsetup{justification=centering}
  \includegraphics[width=.98\linewidth,height=1.95in]{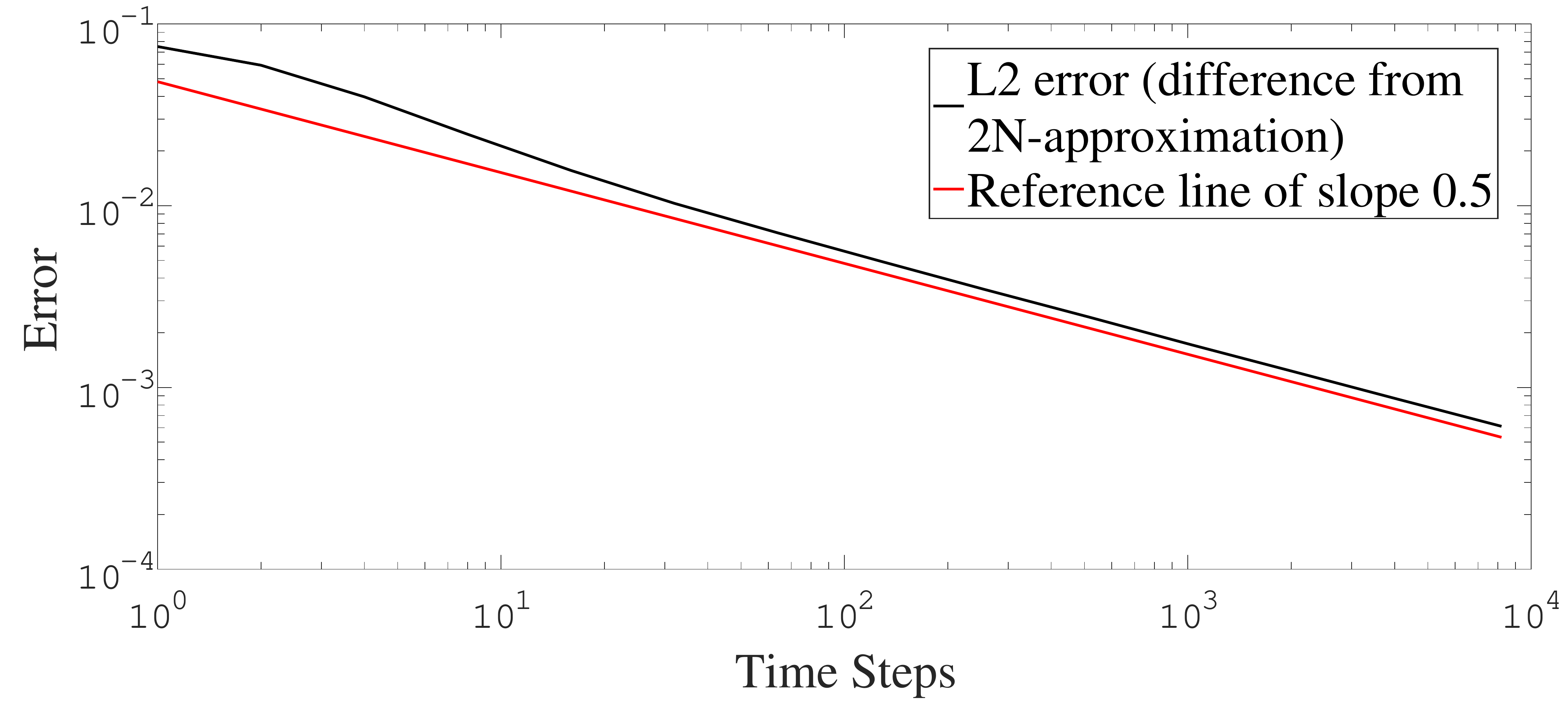}
  \caption{$\nu=2$}
  \label{fig:4d}
\end{subfigure} \\[.5em]
\begin{subfigure}{.5\textwidth}
  \centering
	\captionsetup{justification=centering}
  \includegraphics[width=.98\linewidth,height=1.95in]{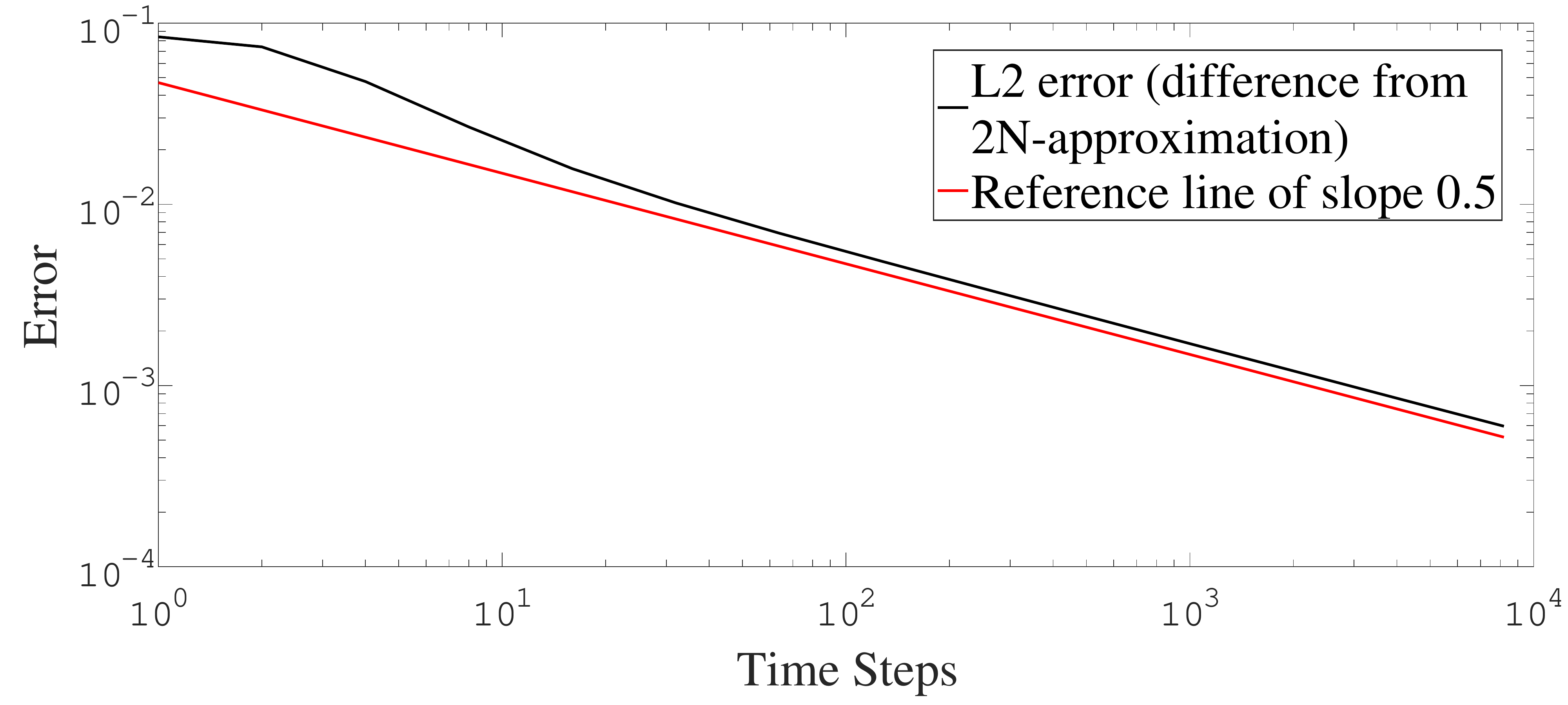}
  \caption{$\nu=4$}
  \label{fig:4e}
\end{subfigure}%
\begin{subfigure}{.5\textwidth}
  \centering
	\captionsetup{justification=centering}
  \includegraphics[width=.98\linewidth,height=1.95in]{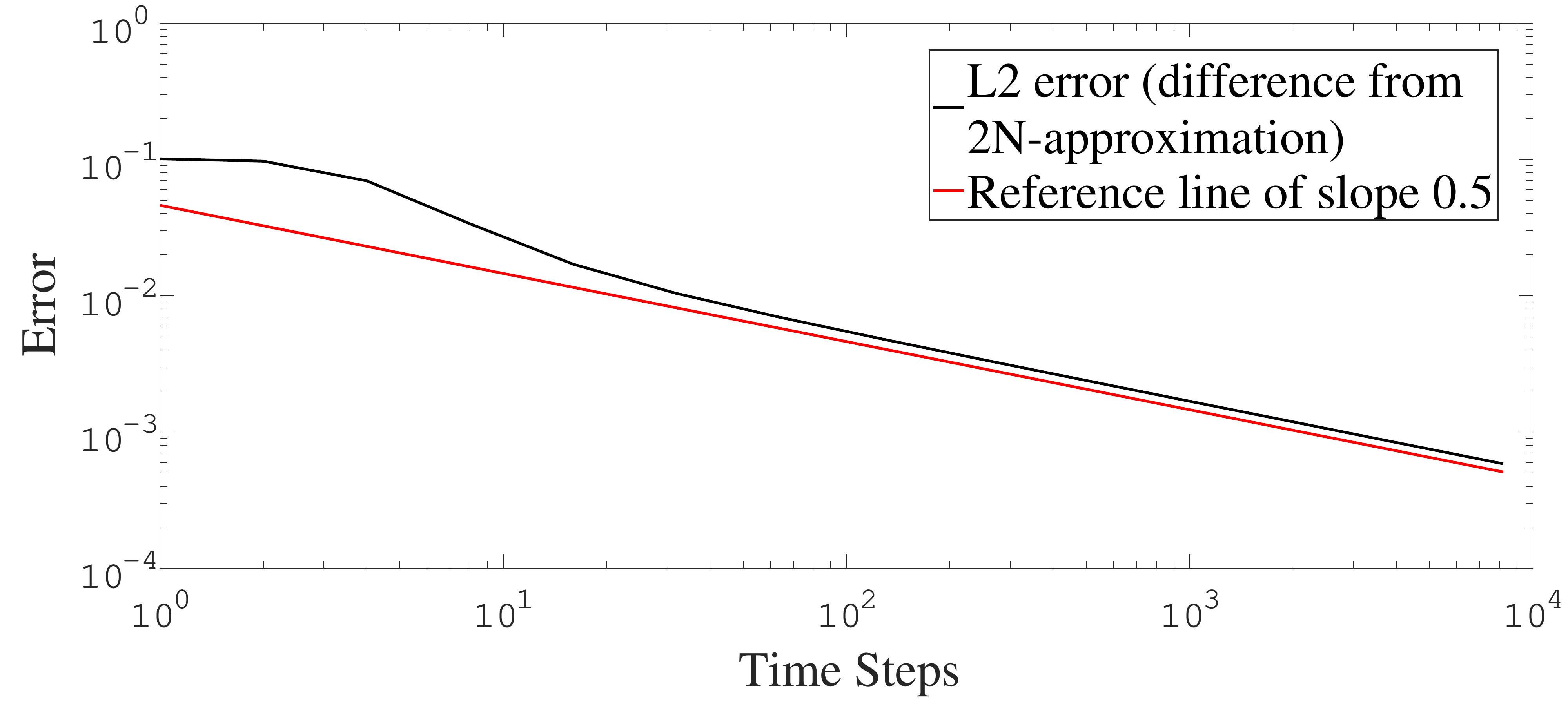}
  \caption{$\nu=8$}
  \label{fig:4f}
\end{subfigure}%
\caption{The $L^{2}$ errors against the number of time steps when $k\in\{0.25,0.5,1,2,4,8\}$ and the other parameters are as defined in \eqref{eq5.1.1}, computed using up to $1.8\!\times\!10^{7}$ Monte Carlo paths (for a relative error less than 10bp).}
\label{fig:4}
\end{figure}

\subsection{Weak convergence}\label{subsec:weak}

We conclude this section with a numerical analysis of the rate of weak convergence. In particular, we consider a European call option with strike $K=0.9$ and time to maturity $T=1$, and assign the same values to the underlying model parameters as in \eqref{eq5.1.1}. In order to observe the asymptotic rate of convergence in a reasonable computational time, we define a new parametric leverage function $\sigma$ with a stronger dependence on the running maximum, namely
\begin{equation}\label{eq5.2.1}
\sigma(t,x,y) = 1+\arctan\big(\log(y)-\log(S_{0})\big).
\end{equation}
Note that this leverage function is bounded, constant in time and spot, and Lipschitz continuous in log-running maximum. Hence, Assumptions \ref{Asm2.1} and \ref{Asm2.2} are satisfied.

In order to establish the strong convergence in $L^{1}$ with order 1/2 (up to a logarithmic factor) -- and hence the weak convergence of the same order -- of the approximation process, we compute the critical time $T^{\scalebox{0.6}{\text{FTE}}}(1)$ from \eqref{eq2.3.9}. A straightforward technical analysis of the leverage function yields $\sigma_{max}=2.571$, $C_{\sigma,x}=0$ and $C_{\sigma,m}=1$. Therefore, we obtain $T^{\scalebox{0.6}{\text{FTE}}}(1)=38.92$, which is greater than $T=1$, and hence all conditions in the statement of Theorem \ref{Thm2.7} are satisfied.

Next, we study the weak error
\begin{equation}\label{eq5.2.2}
\varepsilon_{\scalebox{0.6}{W}}(N) = \big|\E\big[f(S_{T})\big]-\E\big[f(\bar{S}_{T\hspace{-1pt},\hspace{1pt}N})\big]\big|,
\end{equation}
where $f(S)=(S-K)^{+}$ is the European call option payoff. We use Proposition \ref{Prop5.2} and estimate as proxy the difference between the values of the approximated call price corresponding to $N$ time steps ($\bar{S}_{T\hspace{-1pt},\hspace{1pt}N}$) and $2N$ time steps ($\bar{S}_{T\hspace{-1pt},\hspace{1pt}2N}$). The proof is similar to that of Proposition \ref{Prop5.1} and is thus omitted.

\begin{proposition}\label{Prop5.2}
Let $T>0$ and $f:\mathcal{C}(\RR_{+})\rightarrow\RR_{+}$, and suppose that
\begin{equation}\label{eq5.2.3}
\lim_{N\to\infty}\E\big[f(\bar{S}_{T\hspace{-1pt},\hspace{1pt}N})\big] = \E\big[f(S_{T})\big].
\end{equation}
Then, for any $\alpha>0$ and $\beta\geq0$,
\begin{align}\label{eq5.2.4}
&\big|\E\big[f(S_{T})\big]-\E\big[f(\bar{S}_{T\hspace{-1pt},\hspace{1pt}N})\big]\big| = \mathcal{O}\left(\frac{\big(\log(2N)\big)^{\beta}}{N^{\alpha}}\right) \nonumber\\[1pt]
\hspace{1em}\Leftrightarrow\hspace{1em} &\big|\E\big[f(\bar{S}_{T\hspace{-1pt},\hspace{1pt}N})\big]-\E\big[f(\bar{S}_{T\hspace{-1pt},\hspace{1pt}2N})\big]\big| = \mathcal{O}\left(\frac{\big(\log(2N)\big)^{\beta}}{N^{\alpha}}\right).
\end{align}
\end{proposition}

In order to improve the weak convergence rate, we employ Brownian bridge interpolation. Given the approximated log-spot process at two subsequent time nodes, $\bar{x}_{t_{n}}$ and $\bar{x}_{t_{n+1}}$, instead of taking the maximum over a piecewise linear interpolation as in \eqref{eq2.2.10b}, we simulate the maximum of the interpolating Brownian bridge, i.e.,
\begin{equation}\label{eq5.2.10}
\hat{m}_{[t_{n},t_{n+1}]} = \frac{1}{2}\Big[\bar{x}_{t_{n+1}} + \bar{x}_{t_{n}} + \sqrt{(\bar{x}_{t_{n+1}}-\bar{x}_{t_{n}})^{2}-2\sigma^{2}\big(t_{n},e^{\bar{x}_{t_{n}}},e^{\bar{m}_{t_{n}}}\big)\bar{v}_{t_{n}}\delta t\log(U_{n})}\hspace{1.5pt}\Big],
\end{equation}
where $(U_{n})_{0\leq n\leq N-1}$ are independent $\mathcal{U}[0,1]$ random variables, and update the running maximum via
\begin{equation}\label{eq5.2.11}
\bar{m}_{t_{n+1}} = \max\big\{\bar{m}_{t_{n}},\hat{m}_{[t_{n},t_{n+1}]}\big\},\hspace{.75em} \bar{m}_{0} = x_{0}.
\end{equation}

Finally, the data in Figure \ref{fig:5} suggest an empirical weak convergence order of 1/2 with piecewise linear interpolation and an order of 1 with Brownian bridge interpolation, as expected.
\begin{figure}[htb]
\begin{center}
\captionsetup{justification=centering}
\includegraphics[width=1.0\linewidth,height=3.4in]{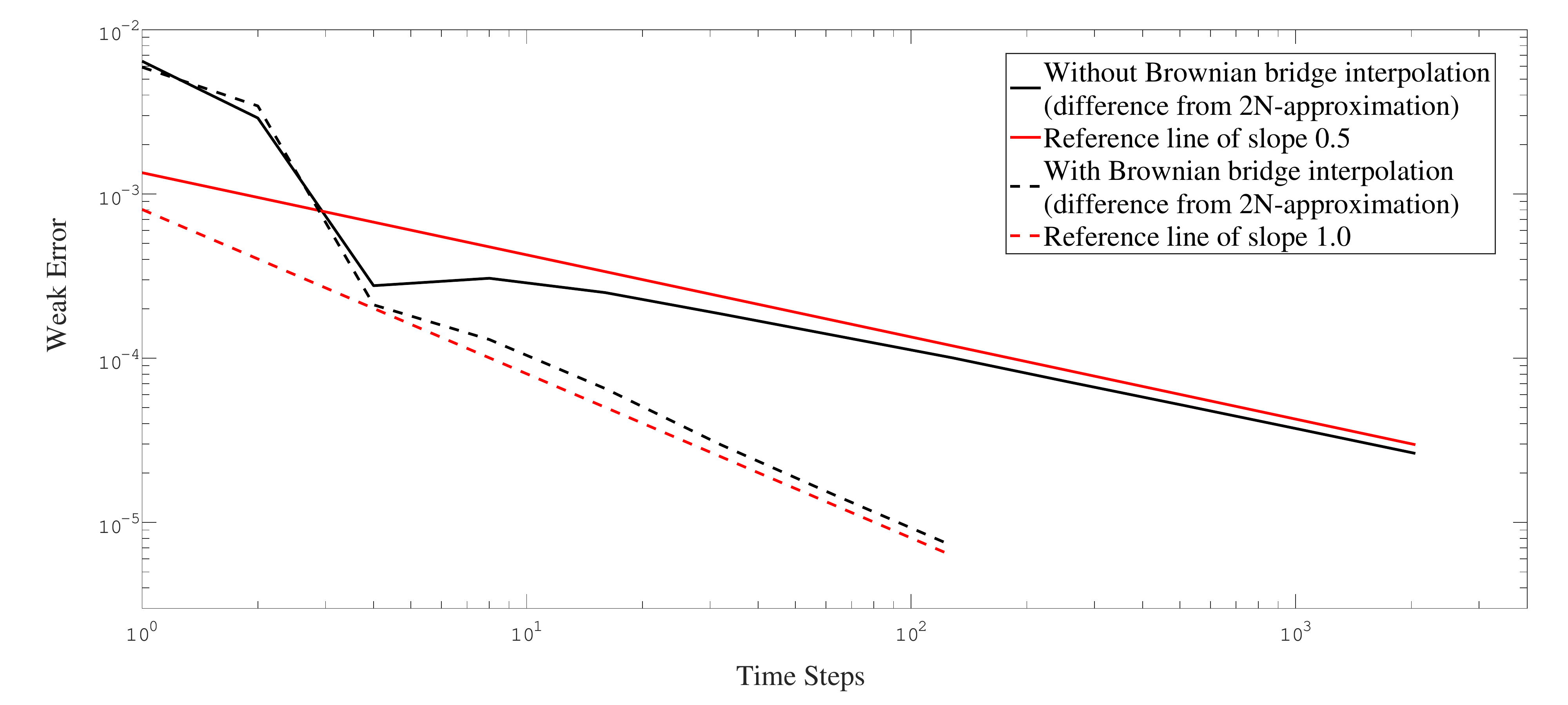}
\end{center}
\caption{The weak errors for a European call payoff (with and without Brownian bridge interpolation) against the number of time steps when the parameters are as defined in \eqref{eq5.1.1} and the strike is $K=0.9$, computed using up to $3.2\!\times\!10^{9}$ Monte Carlo paths (for a relative error less than 1\%).}
\label{fig:5}
\end{figure}

\section{Conclusions}\label{sec:conclusion}

The efficient pricing and hedging of vanilla and exotic options requires an adequate model that takes into account both the local and the stochastic features of the volatility dynamics. In this paper, we have studied a stochastic path-dependent volatility model together with a simple and efficient Monte Carlo simulation scheme. We have made some realistic model assumptions and established, up to a critical time, the strong convergence in $L^{p}$ with order 1/2 up to a logarithmic factor of the Euler approximation. In particular, this enables the use of multilevel simulation, as in \cite{Giles:2009}, with substantial efficiency improvements for the estimation of expected financial payoffs. Inevitably, this work also raises some questions, such as whether we can relax the condition on the stochastic volatility parameters and still deduce similar convergence properties of the scheme, as suggested by our numerical results.

\titleformat{\section}{\large\bfseries}{\appendixname~\thesection .}{0.5em}{}
\begin{appendices}

\section{Proof of Proposition \ref{Prop2.6}}\label{sec:aux1}

First, we show that Assumption \ref{Asm2.1} holds. Using \eqref{eq2.1.4}, \eqref{eq2.1.6} and the triangle inequality, we find that
\begin{align}\label{eq4.1.9}
\left|\mu(t,x,y)\right| &\leq \left|\mu(0,S_{min},S_{min})\right| + \left|\mu(t,S_{min}\vee x\wedge S_{max},S_{min}\vee y\wedge S_{max})-\mu(0,S_{min},S_{min})\right| \nonumber\\[3pt]
&\leq \left|\mu(0,S_{min},S_{min})\right| + C_{\mu,t}\Ind_{t\neq0} +\hspace{2pt} C_{\mu,S}\left|\left(S_{min}\vee x\wedge S_{max}\right)-S_{min}\right| \nonumber\\[2pt]
&+ C_{\mu,M}\left|\left(S_{min}\vee y\wedge S_{max}\right)-S_{min}\right| \nonumber\\[3pt]
&\leq \left|\mu(0,S_{min},S_{min})\right| + C_{\mu,t} + \left(C_{\mu,S}+C_{\mu,M}\right)\left(S_{max}-S_{min}\right).
\end{align}
Similarly, using \eqref{eq2.1.5} and \eqref{eq2.1.7}, we find that
\begin{equation}\label{eq4.1.10}
\sigma(t,x,y) \leq \sigma(0,S_{min},S_{min}) + C_{\sigma,t}\sqrt{T} + \sum_{j=1}^{N_{T}}{C_{\sigma,t,j}} + \left(C_{\sigma,S}+C_{\sigma,M}\right)\left(S_{max}-S_{min}\right).
\end{equation}
Second, we show that Assumption \ref{Asm2.2} holds. Using \eqref{eq2.1.4} and \eqref{eq2.1.6}, we find that
\begin{align}\label{eq4.1.11}
\left|\mu(t_{1},x_{1},y_{1})-\mu(t_{2},x_{2},y_{2})\right| &\leq C_{\mu,t}\Ind_{t_{1}\neq t_{2}} +\hspace{2pt} C_{\mu,S}\left|\left(S_{min}\vee x_{1}\wedge S_{max}\right)-\left(S_{min}\vee x_{2}\wedge S_{max}\right)\right| \nonumber\\[2pt]
&+ C_{\mu,M}\left|\left(S_{min}\vee y_{1}\wedge S_{max}\right)-\left(S_{min}\vee y_{2}\wedge S_{max}\right)\right|.
\end{align}
For convenience, define the function $f_{s}:\RR_{+}\rightarrow\RR_{+}$ given by
\begin{equation}\label{eq4.1.12}
f_{s}(a) = S_{min}\vee a\wedge S_{max}.
\end{equation}
From the Mean-Value Theorem, we know that, for $a_{1,2}\in\{x_{1,2},y_{1,2}\}$,
\begin{equation}\label{eq4.1.13}
\left|f_{s}(a_{1})-f_{s}(a_{2})\right| \leq S_{max}\left|\log\left(f_{s}(a_{1})\right)-\log\left(f_{s}(a_{2})\right)\right|.
\end{equation}
Furthermore, since $f_{s}(a)$ is increasing and $a^{-1}f_{s}(a)$ is decreasing, we have that
\begin{equation}\label{eq4.1.14}
\left|\log\left(f_{s}(a_{1})\right)-\log\left(f_{s}(a_{2})\right)\right| = \log\left(\frac{f_{s}(a_{1}\vee a_{2})}{f_{s}(a_{1}\wedge a_{2})}\right) \leq \log\left(\frac{a_{1}\vee a_{2}}{a_{1}\wedge a_{2}}\right) = \left|\log(a_{1})-\log(a_{2})\right|.
\end{equation}
Combining \eqref{eq4.1.11}, \eqref{eq4.1.13} and \eqref{eq4.1.14}, we deduce that
\begin{align}\label{eq4.1.15}
\left|\mu(t_{1},x_{1},y_{1})-\mu(t_{2},x_{2},y_{2})\right| &\leq C_{\mu,t}\Ind_{t_{1}\neq t_{2}} +\hspace{2pt} C_{\mu,S}S_{max}\left|\log(x_{1})-\log(x_{2})\right| \nonumber\\[2pt]
&+ C_{\mu,M}S_{max}\left|\log(y_{1})-\log(y_{2})\right|.
\end{align}
Similarly, using \eqref{eq2.1.5} and \eqref{eq2.1.7}, we deduce that
\begin{align}\label{eq4.1.16}
\left|\sigma(t_{1},x_{1},y_{1})-\sigma(t_{2},x_{2},y_{2})\right| &\leq C_{\sigma,t}\sqrt{\left|t_{1}-t_{2}\right|} + \sum_{j=1}^{N_{T}}{C_{\sigma,t,j}\Ind_{t_{1}\wedge\hspace{.5pt}t_{2}\hspace{.5pt}<\frac{jT}{N_{T}}\leq\hspace{.5pt}t_{1}\vee\hspace{.5pt}t_{2}}} \nonumber\\[0pt]
&+ C_{\sigma,S}S_{max}\left|\log(x_{1})-\log(x_{2})\right| + C_{\sigma,M}S_{max}\left|\log(y_{1})-\log(y_{2})\right|,
\end{align}
which concludes the proof. \qed

\section{Proof of Lemma \ref{Lem4.3}}\label{sec:aux2}

Since $N$ is a multiple of $N_{T}$, using \eqref{eq2.0.3}, \eqref{eq2.0.4} and the triangle inequality, we get
\begin{align}\label{eq4.0.3}
\left|\mu\big(u,S_{u},M_{u}\big)-\mu\big(u,\bar{S}_{\bar{u}},\bar{M}_{\bar{u}}\big)\right| &\leq C_{\mu,x}\left|x_{u}-x_{\bar{u}}\right| + C_{\mu,x}\left|x_{\bar{u}}-\bar{x}_{\bar{u}}\right| \nonumber\\[3pt]
&+ C_{\mu,m}\left|m_{u}-m_{\bar{u}}\right| + C_{\mu,m}\left|m_{\bar{u}}-\bar{m}_{\bar{u}}\right|
\end{align}
and
\begin{align}\label{eq4.0.4}
\left|\sigma\big(u,S_{u},M_{u}\big)-\sigma\big(\bar{u},\bar{S}_{\bar{u}},\bar{M}_{\bar{u}}\big)\right| &\leq C_{\sigma,t}\sqrt{\delta t} + C_{\sigma,x}\left|x_{u}-x_{\bar{u}}\right| + C_{\sigma,x}\left|x_{\bar{u}}-\bar{x}_{\bar{u}}\right| \nonumber\\[3pt]
&+ C_{\sigma,m}\left|m_{u}-m_{\bar{u}}\right| + C_{\sigma,m}\left|m_{\bar{u}}-\bar{m}_{\bar{u}}\right|.
\end{align}
First, we clearly have that
\begin{equation}\label{eq4.0.5}
\left|x_{u}-x_{\bar{u}}\right| \leq \sup_{t\in[0,u]}\left|x_{t}-x_{\bar{t}}\right| \hspace{1em}\text{ and }\hspace{1em} \left|x_{\bar{u}}-\bar{x}_{\bar{u}}\right| \leq \sup_{t\in[0,u]}\left|x_{t}-\bar{x}_{t}\right|.
\end{equation}
Second, note that
\begin{equation}\label{eq4.0.6}
\left|m_{u}-m_{\bar{u}}\right| = \sup_{t\in[0,u]}x_{t} - \sup_{t\in[0,\bar{u}]}x_{t} \leq \sup_{t\in[0,u]}\left(x_{\bar{t}} + \left|x_{t}-x_{\bar{t}}\right|\right) - \sup_{t\in[0,u]}x_{\bar{t}} \leq \sup_{t\in[0,u]}\left|x_{t}-x_{\bar{t}}\right|.
\end{equation}
Third, note that
\begin{align}\label{eq4.0.7}
\left|m_{\bar{u}}-\bar{m}_{\bar{u}}\right| &= |\hspace{-2pt}\sup_{t\in[0,\bar{u}]}x_{t} - \sup_{t\in[0,\bar{u}]}\bar{x}_{\bar{t}}| \leq \sup_{t\in[0,\bar{u}]}\left|x_{t}-\bar{x}_{\bar{t}}\right| \leq \sup_{t\in[0,\bar{u}]}\left|x_{t}-x_{\bar{t}}\right| + \sup_{t\in[0,\bar{u}]}\left|x_{\bar{t}}-\bar{x}_{\bar{t}}\right| \nonumber\\[2pt]
&\leq \sup_{t\in[0,u]}\left|x_{t}-x_{\bar{t}}\right| + \sup_{t\in[0,u]}\left|x_{t}-\bar{x}_{t}\right|.
\end{align}
Substituting back into \eqref{eq4.0.3} and \eqref{eq4.0.4} with the upper bounds derived in \eqref{eq4.0.5} -- \eqref{eq4.0.7} leads to the conclusion. \qed

\section{Proof of Lemma \ref{Lem4.6}}\label{sec:aux3}

(1) The argument follows that of Proposition 3.12 in \cite{Cozma:2016a}. Fix $p>1$ and note that
\begin{equation}\label{eqC.1.1}
S_{t}^{p} \leq S_{0}^{p}\exp\bigg\{p\mu_{max}t - \frac{p}{2}\int_{0}^{t}{\sigma^{2}\big(u,S_{u},M_{u}\big)v_{u}\,du} + p\int_{0}^{t}{\sigma\big(u,S_{u},M_{u}\big)\sqrt{v_{u}}\,dW^{s}_{u}}\bigg\}.
\end{equation}
Consider the H\"older pair $(q_{1},q_{2})$ given by
\begin{equation}\label{eqC.1.2}
q_{1} = 1+\sqrt{\frac{p-1}{p}} \hspace{5pt}\text{ and }\hspace{5pt} q_{2} = 1+\sqrt{\frac{p}{p-1}}\hspace{1pt}.
\end{equation}
Next, define the stochastic process
\begin{equation}\label{eqC.1.3}
Y_{t} = pq_{1}\int_{0}^{t}{\sigma(u,S_{u},M_{u})\sqrt{v_{u}}\,dW_{u}^{s}}
\end{equation}
with quadratic variation
\begin{equation}
\langle Y\rangle_{t} = p^{2}q_{1}^{2}\int_{0}^{t}{\sigma^{2}(u,S_{u},M_{u})v_{u}\,du}.
\end{equation}
Taking expectations in \eqref{eqC.1.1}, we deduce that
\begin{align}\label{eqC.1.4}
\E\big[S_{t}^{p}\big] \leq S_{0}^{p}e^{p \mu_{max}t}\E\bigg[\exp\bigg\{\frac{1}{q_{1}}\left[Y_{t} - \frac{1}{2}\langle Y\rangle_{t}\right] + \frac{1}{2}\hspace{1pt}p(pq_{1}-1)\int_{0}^{t}{\sigma^{2}(u,S_{u},M_{u})v_{u}\,du}\bigg\}\bigg].
\end{align}
Applying H\"older's inequality with the pair from \eqref{eqC.1.2} and taking the supremum over $[0,T]$ yields
\begin{align}\label{eqC.1.5}
\sup_{t \in [0,T]}\E\big[S_{t}^{p}\big] &\leq S_{0}^{p}e^{p \mu_{max}T}\sup_{t \in [0,T]}\E\bigg[\exp\bigg\{Y_{t} - \frac{1}{2}\langle Y\rangle_{t}\bigg\}\bigg]^{\frac{1}{q_{1}}} \nonumber \\[0pt]
&\times \E\bigg[\exp\bigg\{\frac{1}{2}\hspace{1pt}pq_{2}\big(pq_{1}-1\big)\sigma_{max}^{2}\int_{0}^{T}{v_{u}\,du}\bigg\}\bigg]^{\frac{1}{q_{2}}}.
\end{align}
The stochastic exponential is a martingale if Novikov's condition is satisfied, and hence if
\begin{equation}\label{eqC.1.6}
\E\bigg[\exp\bigg\{\,\frac{1}{2}\langle Y\rangle_{T}\bigg\}\bigg] \leq \E\bigg[\exp\bigg\{\frac{1}{2}\hspace{1pt}p^{2}q_{1}^{2}\sigma_{max}^{2}\int_{0}^{T}{v_{u}\,du}\bigg\}\bigg] < \infty.
\end{equation}
The finiteness of the two expectations in \eqref{eqC.1.5} follows from Lemma \ref{Lem3.2}.

\noindent
(2) The argument follows that of Proposition 3.13 in \cite{Cozma:2016a}. Fix $p>1$ and note that
\begin{equation}\label{eqC.2.1}
\bar{S}_{t}^{p} \leq S_{0}^{p}\exp\bigg\{p\mu_{max}t - \frac{p}{2}\int_{0}^{t}{\sigma^{2}\big(\bar{u},\bar{S}_{\bar{u}},\bar{M}_{\bar{u}}\big)\bar{v}_{u}\,du} + p\int_{0}^{t}{\sigma\big(\bar{u},\bar{S}_{\bar{u}},\bar{M}_{\bar{u}}\big)\sqrt{\bar{v}_{u}}\,dW^{s}_{u}}\bigg\}.
\end{equation}
Henceforth, we argue as before and use Lemmas \ref{Lem3.4} and \ref{Lem3.7}. \qed

\section{Proof of Proposition \ref{Prop5.1}}\label{sec:aux4}

Suppose that there exists a constant $C>0$ such that, for all $N\geq1$,
\begin{equation}\label{eq5.1.9}
\E\Big[\big|S_{T}-\bar{S}_{T\hspace{-1pt},\hspace{1pt}N}\big|^{p}\Big]^{\frac{1}{p}} \leq C\hspace{1pt}\frac{\big(\log(2N)\big)^{\beta}}{N^{\alpha}}\hspace{1pt}.
\end{equation}
Using this upper bound and the triangle inequality yields
\begin{align}\label{eq5.1.10}
\E\Big[\big|\bar{S}_{T\hspace{-1pt},\hspace{1pt}N}-\bar{S}_{T\hspace{-1pt},\hspace{1pt}2N}\big|^{p}\Big]^{\frac{1}{p}} &\leq
2^{1-{\frac{1}{p}}}\E\Big[\big|S_{T}-\bar{S}_{T\hspace{-1pt},\hspace{1pt}N}\big|^{p}\Big]^{\frac{1}{p}} + 2^{1-{\frac{1}{p}}}\E\Big[\big|S_{T}-\bar{S}_{T\hspace{-1pt},\hspace{1pt}2N}\big|^{p}\Big]^{\frac{1}{p}} \nonumber\\[1pt]
&\leq 2^{1-{\frac{1}{p}}}\big(1+2^{\beta-\alpha}\big)C\hspace{1pt}\frac{\big(\log(2N)\big)^{\beta}}{N^{\alpha}}\hspace{1pt}.
\end{align}
Conversely, suppose that there exists a constant $C_{1}>0$ such that, for all $N\geq1$,
\begin{equation}\label{eq5.1.11}
\E\Big[\big|\bar{S}_{T\hspace{-1pt},\hspace{1pt}N}-\bar{S}_{T\hspace{-1pt},\hspace{1pt}2N}\big|^{p}\Big]^{\frac{1}{p}} \leq C_{1}\hspace{1pt}\frac{\big(\log(2N)\big)^{\beta}}{N^{\alpha}}\hspace{1pt}.
\end{equation}
Fix any $1<\gamma<\frac{\eta p}{p-1}$ and define the sequence $(a_{i})_{i\geq0}$ given by
\begin{equation}\label{eq5.1.12}
a_{i} = (i+1)^{-\gamma}.
\end{equation}
For any $l\in\mathbb{N}\cup\{0\}$ and $x_{0},x_{1},\hdots,x_{l}\geq0$, H\"older's inequality yields
\begin{equation}\label{eq5.1.13}
\Bigg(\sum_{i=0}^{l}{x_{i}}\Bigg)^{p} \leq \Bigg(\sum_{i=0}^{l}{a_{i}^{1-p}x_{i}^{p}}\Bigg)\Bigg(\sum_{i=0}^{l}{a_{i}}\Bigg)^{p-1}.
\end{equation}
Furthermore,
\begin{equation}\label{eq5.1.14}
\sum_{i=0}^{l}{a_{i}} < \zeta(\gamma) < \infty,
\end{equation}
where $\zeta$ is the Riemann zeta function. Using the triangle inequality, \eqref{eq5.1.13} and \eqref{eq5.1.14}, and then taking expectations, we get
\begin{equation}\label{eq5.1.15}
\E\Big[\big|S_{T}-\bar{S}_{T\hspace{-1pt},\hspace{1pt}N}\big|^{p}\Big] \leq \zeta(\gamma)^{p-1}\sum_{i=0}^{l-1}{a_{i}^{1-p}\E\Big[\big|\bar{S}_{T\hspace{-1pt},\hspace{1pt}2^{i}N}-\bar{S}_{T\hspace{-1pt},\hspace{1pt}2^{i+1}N}\big|^{p}\Big]} + \zeta(\gamma)^{p-1}a_{l}^{1-p}\E\Big[\big|\bar{S}_{T\hspace{-1pt},\hspace{1pt}2^{l}N}-S_{T}\big|^{p}\Big].
\end{equation}
However, we know from \eqref{eq5.1.7} that there exists a constant $C_{2}>0$ such that, for all $N\geq1$,
\begin{equation}\label{eq5.1.16}
\E\Big[\big|S_{T}-\bar{S}_{T\hspace{-1pt},\hspace{1pt}N}\big|^{p}\Big]^{\frac{1}{p}} \leq C_{2}\big(\log(2N)\big)^{-\eta}.
\end{equation}
Substituting back into \eqref{eq5.1.15} with the upper bounds in \eqref{eq5.1.11} and \eqref{eq5.1.16}, we deduce that
\begin{align}\label{eq5.1.17}
\E\Big[\big|S_{T}-\bar{S}_{T\hspace{-1pt},\hspace{1pt}N}\big|^{p}\Big] &\leq C_{1}^{p}\zeta(\gamma)^{p-1}\hspace{1pt}\frac{\big(\log(2N)\big)^{\beta p}}{N^{\alpha p}}\sum_{i=0}^{l-1}{\frac{(i+1)^{\gamma(p-1)+\beta p}}{2^{i\alpha p}}} \nonumber\\[1pt]
&+ C_{2}^{p}\big(\log(2)\big)^{-\eta p}\zeta(\gamma)^{p-1}(l+1)^{\gamma(p-1)-\eta p},
\end{align}
and taking the limit as $l$ goes to infinity leads to
\begin{equation}\label{eq5.1.18}
\E\Big[\big|S_{T}-\bar{S}_{T\hspace{-1pt},\hspace{1pt}N}\big|^{p}\Big] \leq C_{1}^{p}2^{\alpha p}\zeta(\gamma)^{p-1}\hspace{1pt}\frac{\big(\log(2N)\big)^{\beta p}}{N^{\alpha p}}\sum_{n=1}^{\infty}{\frac{n^{\gamma(p-1)+\beta p}}{2^{n\alpha p}}}\hspace{1pt}.
\end{equation}
The conclusion follows from the fact that the series on the right-hand side converges. \qed

\end{appendices}

\bibliographystyle{abbrv}
\bibliography{references}

\end{document}